\documentclass[11pt]{article}

 \usepackage{amsmath,   amssymb,amsfonts,xspace,graphicx,relsize,bm,mathtools,xcolor}
    \usepackage{mathrsfs}
    \usepackage[pagebackref]{hyperref}
    \usepackage{enumerate}
    \usepackage{bm}
 \hypersetup{
	colorlinks,
	linkcolor={blue!100!black},
	citecolor={red!100!black},
}
\usepackage[margin=1in]{geometry}
\usepackage{amsmath,amssymb,amsthm}
\usepackage{graphicx}
\usepackage{comment}
\usepackage[capitalise]{cleveref}
\usepackage{braket}
\def\01{\{0,1\}}
\usepackage{tcolorbox}
\newtheorem{theorem}{Theorem}[section]

\newtheorem{corollary}{Corollary}[theorem]

\newtheorem{definition}{Definition}[section]
\newtheorem{fact}{Fact}[theorem]
\newtheorem{lemma}[theorem]{Lemma}
\newcommand{\E}{\mathbb{E}}

\usepackage{circledsteps}

\usepackage{url}
\def\01{\{0,1\}}

\newcommand{\eps}{\varepsilon}
\newcommand{\ketbra}[2]{|#1\rangle\langle#2|}

\newcommand{\Oh}{\ensuremath{\mathcal{O}}}

\newcommand{\id}{\ensuremath{\mathbb{I}}}

\usepackage{enumitem}

\DeclareMathOperator{\polylog}{polylog}

\newcommand{\Tr}{\mathsf{Tr}}
\newcommand{\LOCC}{\mathsf{LOCC}}
\newcommand{\DIPE}{\textsf{DIPE}}
\newcommand{\GDIPE}{\textsf{GDIPE}}
\newcommand{\YES}{\textsf{YES}}
\newcommand{\Var}{\textsf{Var}}
\newcommand{\NO}{\textsf{NO}}

\DeclareMathOperator{\swap}{SWAP}

\begin{document}

\title{Distributed inner product estimation\\ with limited quantum communication}
\author{
Srinivasan Arunachalam\\[2mm]
IBM Quantum\\
\small Almaden Research Center\\
\small \texttt{Srinivasan.Arunachalam@ibm.com}
\and
Louis Schatzki\\[2mm]
\small   Electrical and Computer Engineering\\
University of Illinois\\
\small \texttt{louisms2@illinois.edu}
}

\maketitle
\begin{abstract}
    We consider the task of distributed inner product estimation when allowed limited quantum communication.  Here, Alice and Bob are given $k$ copies of an unknown $n$-qubit quantum states $\ket{\psi},\ket{\phi}$ respectively. They are allowed to communicate $q$ qubits and unlimited classical communication, and their goal is to estimate $|\langle \psi|\phi\rangle|^2$ up to constant accuracy. We show that $k=\Theta(\sqrt{2^{n-q}})$ copies are essentially necessary and sufficient for this task (extending the work of Anshu, Landau and Liu (STOC'22) who considered the case when $q=0$). Additionally, we consider estimating $|\langle \psi|M|\phi\rangle|^2$, for arbitrary Hermitian $M$. For this task we show that certain norms on $M$ \emph{characterize} the sample complexity of estimating $|\langle \psi|M|\phi\rangle|^2$ when using only classical~communication.
\end{abstract}

\section{Introduction}
The seminal work of Buhrman, Cleve, Watrous and Wolf~\cite{buhrman2001quantum} on \emph{quantum fingerprinting} introduced the so-called \emph{swap test}. This subroutine takes as input one copy of unknown quantum states $\ket{\psi}$ and $\ket{\phi}$ and outputs a bit whose bias equals $|\langle \psi|\phi\rangle|^2$; allowing us to estimate the \emph{inner product} between two unknown quantum states without knowing anything about the states themselves. This subroutine has found applications in several areas in quantum computing from complexity theory, learning theory, entanglement theory, to optimization. Although so widely used, there are fundamental questions about estimating inner products, which is the topic of this~work. 

We consider a natural \emph{distributed} variant of the inner product estimation problem: suppose there are two spatially separated parties, Alice and Bob, who are each given copies of an unknown $d$-dimensional quantum state, $\ket{\psi}$ and $\ket{\phi}$ respectively, with the goal of computing $|\langle \psi|M|\phi\rangle|^2$ for some Hermitian operator $M$ using as few copies as possible. We call this the \emph{generalized distributed inner product estimation} problem ($\GDIPE$). The setting of $M=\id$ in $\GDIPE$ (which we refer to as \emph{distributed inner product estimation} $\DIPE$) is one of the important steps in cross-platform verification proposed in~\cite{elben2020cross}, where the goal is to test if two quantum computers (say built on different platforms) are preparing the same quantum state; a fundamental problem for near-term devices. A natural constraint in these papers is how the quantum computers communicate with one another, and in these works they consider either only classical communication or limited quantum communication (the latter being the focus of our work). Apart from verification, on a theoretical level, understanding the sample complexity of a task as fundamental as inner product estimation in a distributed manner is a natural~question. 

There are two known techniques to solve $\GDIPE$ under different settings: $(i)$ If allowed quantum communication, then Alice can  send a few copies of the $d$-dimensional quantum state over to Bob, who then performs a variant of the swap test to estimate $|\langle \psi|M|\phi\rangle|^2$; $(ii)$ if allowed only classical communication, i.e., Alice and Bob can only send classical messages, then a recent work of Anshu et al.~\cite{anshu2022distributed} showed that $\Theta(\sqrt{d})$ copies of $\ket{\psi},\ket{\phi}$ are necessary and sufficient in order to estimate $|\langle \psi|\phi\rangle|^2$. This setting is often referred to as \emph{local operations and classical communication} ($\LOCC$) and has received a lot of attention in quantum information theory~\cite{chitambar2014everything,bennett1999quantum,fan2004distinguishability}. Although $(ii)$ is surprising in that, the sample complexity is quadratically better than full-state tomography, the exponential  nature of the complexity (since we typically work with $d=2^n$, where $n$ is the number of qubits) seems rather large in practice. Their work raises two natural questions, which will be the focus here:
\begin{enumerate}
    \item  With the advent of near-term quantum devices, sending the full-quantum state (as in the protocol $(i)$ above) might be hard, but it is plausible that one could send a \emph{few} qubits of quantum message. If this were possible, then what is the complexity of $\DIPE$ when allowed $q$-qubit messages? 
    \item The work of Anshu et al.~\cite{anshu2022distributed} showed how to solve $\GDIPE$ when $M=\id$, but other natural properties of quantum states are captured by letting $M$ be an arbitrary operator, for example,  two-point correlation functions, entanglement entropy of
small subsystems,  expectation values of observables, projector onto a subspace or perhaps a Pauli string; and in these case $M$ need not be $\id$. In this case, understanding the  complexity of $\GDIPE$ for arbitrary $M$~is~relevant.
\end{enumerate}
These questions were explicitly raised by Anshu et al.~\cite{anshu2022distributed} (and as open question 21 in~\cite{anshu2024survey}) and in this work we answer both.

\subsection{Results}
Our first result considers the setting where Alice and Bob may communicate a few qubits in an interactive protocol but not necessarily enough to teleport copies of their states. For simplicity assume that Alice and Bob each have $n$-qubit states and can transmit in total $q$ qubits of communication. As we mentioned above, the sample complexity when $q=0$ was shown to be $\Theta\left(\sqrt{2^n}\right)$ by~\cite{anshu2022distributed} and when $q=n$, the sample complexity is $O(1)$. Is there some saving in sample complexity when allowed $q$ qubits of communication? Our first result shows a smooth interpolation between both these settings.

\begin{theorem}[Informal]
    Suppose Alice and Bob are given $k$ copies of $n$-qubit states $\ket{\psi}$ and $\ket{\phi}$ respectively and can communication $\Theta(q)$ qubits along with performing arbitrary $\LOCC$. Then it is necessary and sufficient to obtain
$     k=\Theta\left( \sqrt{2^{n-q}} \right)
 $   copies of their states to estimate $\vert \langle \phi \vert \psi \rangle \vert^2$ up to constant accuracy with high~probability.\footnote{We will state the constants more clearly when proving this theorem.}
\end{theorem}
The lower bound on sample complexity of $k = \Omega\left( 2^{(n-q)/2} \right)$ shows that quantum communication may help. However, to reach sub-exponential scaling, a large amount of quantum communication is necessary i.e., unless $q = n/ \polylog(n)$, our sample complexity lower bound is still exponential in $n$. Informally, entanglement helps, but not too much.

Our next result concerns the sample complexity, under $\LOCC$, of $\GDIPE$, i.e., for arbitrary Hermitian operators $M$. A natural possibility when estimating $|\langle \psi|M|\phi\rangle|^2$ is that, perhaps Alice and Bob only need to confirm that their states have a large overlap in one of several smaller subspaces ``within $M$". Then, it is conceivable that this task would be much easier than full inner product estimation (i.e., when $M=\id$).  We show that this is indeed the case!  Let $M_\varepsilon$ be $M$ restricted to subspaces with eigenvalue of magnitude at least $\varepsilon/2$. We prove the following near-optimal results for general $\GDIPE$.
\begin{theorem}[Informal]\label{thm:blf_ub_informal}
      For every $M$ with $\|M\|\leq 1$, with only local operations and classical communication, it is sufficient to obtain \[k=\Oh(\max \{1/\eps^2, \Vert M_\eps \Vert_2/\eps\})\] copies of $\ket{\psi},\ket{\phi}$  to estimate $| \langle \phi | M | \psi \rangle |^2$ to error $\varepsilon$ with high probability and it necessary to obtain \[k=\Omega(\max \{1/\eps^2, \Vert M_\eps \Vert_2/\sqrt{\eps}\})\] copies to produce an $\varepsilon$-approximation to $| \langle \phi | M | \psi \rangle |^2$.
\end{theorem}

 Interestingly, this implies that $M$ having both positive and negative eigenvalues may not render estimation harder than definite operators. Take Pauli strings as an example, an important set of operators in quantum information. These have eigenspaces of eigenvalues $\pm 1$ each of dimension $2^{n-1}$. Then, in evaluating $\vert \langle \phi \vert M \vert \psi \rangle \vert^2$ there may be cancellations between components of the states lying in eigenspaces of different signs. However, our results show that estimating such an $M$ is no harder than estimating $\id$, which is positive definite and basis independent. Another interesting observation is that the ambient dimension of the entire Hilbert space, $d$, does not factor into these bounds. Thus, embedding some $M$ into an arbitrarily large Hilbert space does not change the sample complexity.
 
 While our upper and lower bounds do not exactly match, this is partially due to choice of presentation. The lower bound we prove is actually $\Omega( \Vert M_\varepsilon \Vert_1 / \varepsilon\sqrt{d_\varepsilon} )$, where $d_\varepsilon$ is the dimension of the support of $M_\varepsilon$. Converting $1$ to $2$ norm yields the lower bound as presented above. However, in some cases this more precise bound is tight. Take, for example, $M$ a projector onto a subspace or a Pauli string. Then, we obtain the lower bound $\Omega(\Vert M_\varepsilon \Vert_2 /\varepsilon)$, \emph{exactly} matching the upper bound.

\paragraph{Proof sketch.}
In~\cite{anshu2022distributed}, they prove their lower bound using a very interesting connection to quantum cloning: in particular, they show that if one could solve a decision-version of $\DIPE$, then one could have cloned the states ``well-enough". One of the main challenges in proving our lower bound was that the cloning-based lower bounds do not work when allowed even few qubits of communication in the cloning channel (which is the setting in which we'd like lower bounds). Instead, here we use the notion of \emph{robustness of entanglement}, a concept from entanglement theory to split a protocol with a quantum channel into a sum of LOCC protocols. Proving this is a small technical lemma, but once we have this, we can invoke the lower bound of~\cite{anshu2022distributed} to obtain our overall lower bound of $\Omega\left(\sqrt{2^{n-q}}\right)$. In order to prove our upper bound, we use ideas from the celebrated Johnson-Lindenstrauss lemma. Our main idea is to perform projections onto random subspaces which maintain the inner product between their states with high probability. Further, this protocol is completely agnostic to the states they started with. Using classical communication, they determine where these subspace projections have collided. Since the states now live in smaller dimensional subspaces, a smaller amount of entanglement is needed to perform teleportation. After teleporting, Bob performs a $\swap$ test to estimate the inner product between these projected pairs. With some analysis, we are able to show that the overall sample complexity for this scales as $\sqrt{2^{n-q}}$ as well, matching our lower bound.

For estimating bilinear forms, our protocol is fairly similar to the one in~\cite{anshu2022distributed} except that we need to deal with some technicalities due to the Hermitian operator $M$. The key observation we first make is that $\vert \langle \phi \vert M \vert \psi \rangle \vert^2$ and $\vert \langle \phi \vert M_\varepsilon \vert \psi \rangle \vert^2$ can only differ by at most~$\varepsilon/2$. So it suffices for Alice and Bob to restrict themselves to the support of $M_\varepsilon$. Then, they can simultaneously measure the magnitude of the components of their states in this subspace as well as the weighted inner product of said components. The calculations here are similar to the one in~\cite{anshu2022distributed} wherein one needs to compute the mean variance of the estimator, but a little bit more technical in order to bound all the terms in terms of $\|M_\varepsilon\|_2$. The lower bound follows from realizing that $M_\varepsilon$ when restricted to its support, is not too far off from identity and then we can invoke the lower bounds for inner product estimation can be applied. We remark that the lower bound here is more subtle than the one in~\cite{anshu2022distributed} since one needs to project onto the eigenspaces of $M$ carefully in order to get the optimal dependence on~$\|M_\varepsilon\|_2$.

\subsection{Related works and open questions}
Our work answers two open questions raised by Anshu et al.~\cite{anshu2022distributed}. There have been a few more papers since their work: Hinsche et al.~\cite{hinsche2024efficient} look at specific instantiations of $\DIPE$ (when the protocol used by Alice and Bob are Pauli sampling and the quantum states are structured) and prove some positive and negative results to this end. Chen et al.~\cite{chen2023unitarity} extend the work of~\cite{anshu2022distributed} by proving better bounds for non-trace-preserving quantum channels, Chen et al.~\cite{chen2024local} also looked at quantum property testing in the $\LOCC$~model. Concurrent to our work, Gong et al.~\cite{jonasetal} have similar results to ours and we thank them for helpful email exchanges. Distinguishability under LOCC has a rich history in the quantum information. Bennett et al.~\cite{bennett1999quantum} demonstrated that there exist product bases which are easily distinguished, but indistinguishable under LOCC. Enough entanglement renders this task easy as the two parties can simply teleport states. Subsequent works showed that lesser entanglement may suffice for this task~\cite{cohen2008understanding, zhang2016entanglement}. Unlike our framework, these papers concerned themselves with measurements of a single copy and further identifying states from some known ensemble, rather than learning a property.

\paragraph{A natural open question.} A tantalizing open question left open by~\cite{anshu2022distributed} and also this work is the analogue of $\DIPE$ in the mixed case setting. Here the setup is as follows: Alice has copies of $\rho$, Bob has copies of $\sigma$ and they engage in $\LOCC$. They need to distinguish if $\rho=\sigma$ or $\|\rho-\sigma\|_{tr}\geq \varepsilon$, promised one of them is the case. Using the bounds obtained in~\cite{anshu2022distributed},  we have an upper bound of $O(d^2/\varepsilon^4+d^{1.5}/\varepsilon^2)$. As for lower bounds,~\cite{o2015quantum} derived a general lower bound of $\Omega(d)$ even without locality constraints. In~\cite{buadescu2019quantum}, they showed that this lower bound can be met via a highly entangled measurement across all copies of $\rho$ and $\sigma$. If we restricted the measurements to be single-copy measurements and and fix $\sigma=\id/d$, then $\Omega(d^{3/2})$ copies are necessary~\cite{chen2022tight}. Despite several attempts in making progress on this question, we have not been able to obtain a lower bound better than $\Omega(d)$ and believe that the sample complexity of mixed state $\DIPE$ should be $O(d)$. We leave this as an open question.

 \paragraph{Acknowledgements.}  This work was done in part while the author was visiting the Simons Institute for the Theory of Computing, supported by DOE QSA grant \#FP00010905. AS and LS were supported  through the IBM-Illinois Discovery Accelerator Institute. We thank the authors of~\cite{jonasetal} for useful discussions and coordinating the arXiv submission. We thank Anurag Anshu, Yunchao Liu, Vojtech Havlicek, Felix Leditzky, and Eric Chitambar for discussions.

\section{Preliminaries}

\subsection{Quantum States and Measurements}
Pure quantum states are unit vectors in $\mathbb{C}^d$. A qubit is a state in $\mathbb{C}^2$ and an $n$-qubit state lives in $\left(\mathbb{C}^2\right)^{\otimes n} \cong \mathbb{C}^{2^n}$. Mixed states $\rho$, also known as density matrices, are positive semi-definite operators on $\mathbb{C}^d$ with unit trace. Pure states correspond to rank-$1$ mixed states and we will routinely use the notation $\psi$ to refer to the density matrix corresponding to a pure state $\ket{\psi}$. Measurements of quantum systems are described by positive operator valued measures (POVMs), ensembles of positive semi-definite operators $\{M_i\}_i$ such that $\sum_i M_i = \id$. The probability of observing an outcome $i$ is given by $\Tr[M_i \rho]$. By $\Vert M \Vert$ we denote the operator norm of an operator. 

In quantum computation states are manipulated via unitary evolution. That is, $\ket{\psi}$ is mapped to $U\ket{\psi}$ for some unitary $U$. An important unitary we will make repeated usage of is the $\swap$ gate on a bipartite state. This has the action $\swap \ket{\psi}\otimes \ket{\phi} = \ket{\phi}\otimes \ket{\psi}$.

Operators on two Hilbert spaces $\mathcal{H}_A$ and $\mathcal{H}_B$ lie in the tensor product space $\mathcal{B}(\mathcal{H}_A \otimes \mathcal{H}_B)$. A positive semi-definite operator $M$ acting on these tensor product space is said to be separable if it admits a decomposition $M = \sum_i A_i \otimes B_i$, where each $A_i$ and $B_i$ are positive semi-definite. We will be interested in separable POVMs, where each aspect $M_i$ is separable. These are measurements that cannot create entanglement between distributed parties. However, such measurements may still go beyond those achievable with classical communication. With this in mind, one can define local operations and classical communication ($\LOCC$)~\cite{chitambar2014everything}. This can roughly be defined as all protocols where Alice performs some measurement in her lab, communicates the result to Bob, who then performs a measurement in his lab and communicates the result to Alice and so on. However, the set of $\LOCC$ measurements is mathematically unwieldy and it is generally easier to prove results regarding separable measurements.  An $e$-bit is a standard resource of entanglement in quantum information. This is a shared two qubit resource state $\ket{\Phi^+} = \frac{1}{\sqrt{2}}(\ket{00} + \ket{11})$. With $n$ $e$-bits and classical communication, two distributed parties can teleport an $n$ qubit state~\cite{nielsen2010quantum}.

\subsection{Subroutines}

We will make repeated reference to the $\swap$ test, which can be used to estimate the overlap between two states $\ket{\psi}$ and $\ket{\phi}$~\cite{barenco1997stabilization, buhrman2001quantum}. 

\begin{definition}[$\swap$ test]
    Given two mixed states $\rho$ and $\sigma$ and an ancilla qubit initialized in the state $\ket{0}_E$, the $\swap$ test performs the unitary $(H_E\otimes \id_{AB}) (\ketbra{0}{0}\otimes \id_{AB} + \ketbra{1}{1}\otimes \swap_{AB}) (H_E\otimes \id_{AB})$ and measures the ancilla in the computation basis. The measurement probabilities are given~by
    \begin{align*}
        p(0) = \frac{1+\Tr[\rho\sigma]}{2},\ \quad p(1) = \frac{1-\Tr[\rho\sigma]}{2}\ .
    \end{align*}
\end{definition}
 Via standard amplification arguments, $\Oh(1/\varepsilon^2)$ trials are sufficient to estimate $\Tr[\psi \phi]$ to error $\varepsilon$. Using block-encodings, this can be extended to estimating $\vert \langle \phi \vert M \vert \psi  \rangle \vert^2$ for an arbitrary hermitian $M$ such that $\Vert M \Vert \leq 1$. Again, standard amplification arguments show that $\Oh(1/\varepsilon^2)$ trials suffices to estimate $\vert \langle \phi \vert M \vert \psi \rangle\vert^2$.
 We will also use the standard POVM on the symmetric subspace, which we now define. The symmetric group $S_k$ has a natural action on the state space of $\left(\mathbb{C}^d\right)^{\otimes k}$ given~by 
\begin{align}
    \pi \bigotimes_{i=1}^k \ket{\psi_i} & = \bigotimes_{i=1}^k \ket{\psi_{\pi^{-1}(i)}}\ \forall \pi \in S_k\ .
\end{align}
\begin{definition}[Symmetric Subspace]
    The $k$-copy symmetric subspace of $ \mathbb{C}^d$, is given by
    \begin{align*}
        \vee^k \mathbb{C}^d = \left\{\ket{\psi} \in \left(\mathbb{C}^d\right)^{\otimes k}\ \vert \ \pi \ket{\psi} = \ket{\psi} \forall \pi \in S_k \right\}\ .
    \end{align*}
\end{definition}

The symmetric subspace has a natural spanning set given by $\{\ket{\phi}^{\otimes k}\ \vert \ \ket{\phi} \in \mathbb{C}^d\}$~\cite{harrow2013church}. Due to $\vee^k \mathbb{C}^d$ being an irreducible representation of the unitary group (under the action $U^{\otimes k}$), we can define a uniform POVM on the symmetric subspace, which is known as the standard POVM.
\begin{definition}[Standard POVM on $\vee^k \mathbb{C}^d$]\label{def:standard_POVM}
    The standard POVM on $\vee^k \mathbb{C}^d$ is the continuous POVM with elements
    \begin{align}
        \left\{ \binom{d+k-1}{k} \ketbra{\varphi}{\varphi}^{\otimes k} d\varphi \ \big\vert \ \ket{\varphi}\in\mathbb{C}^d \right\}\ .
    \end{align}
\end{definition}

We require one last fact about the symmetric subspace:
\begin{fact}[Haar moments]\label{fact:haar_moments}
    If $\ket{\psi}$ is drawn from the Haar measure on $\mathbb{C}^d$, then
    \begin{align}
        \E_\psi[\psi] & = \frac{\id}{d}\ ,\quad \E_\psi[\psi^{\otimes 2}] = \frac{1}{d(d+1)}(\id + \swap)\ .
    \end{align}
\end{fact}

Lastly, we implicitly use the equivalence between quantum communication and shared entanglement plus classical communication. One direction is immediate: given a quantum communication channel, Alice can simply send one half of an entangled state to produce a shared entangled state. The other direction is witnessed by the well-known teleportation protocol.

\begin{definition}[Quantum Teleportation~\cite{bennett1993teleporting, werner2001all}]
    Let Alice and Bob share a $d$-dimensional maximally entangled state, $
    \frac{1}{\sqrt{d}}\sum_{i=1}^d \ket{ii}$. Then, using classical communication and consuming this shared entangled state, Alice may transmit a state $\ket{\psi} \in \mathbb{C}^d$ to Bob.
\end{definition}
It is important to note that this protocol works even if $\ket{\psi}$ is embedded into an arbitrary large Hilbert space. Say we know that $\ket{\psi}$ is entirely contained in some subspace $W \subseteq \mathbb{C}^d$ of dimension $d'$. Then, a shared maximally entangled state of dimension $d'$ suffices to teleport $\ket{\psi}$. This will be crucial in our upper bound for $\DIPE$.

\section{Estimating Inner Product}

In this section we consider the following task: Alice and Bob each have copies of unknown $d$-dimensional states $\ket{\phi}$ and $\ket{\psi}$ respectively. They may use any amount of classical communication and some restricted quantum communication to estimate $| \langle \phi | \psi \rangle |^2$ using as few samples of $\ket{\psi}$ and $\ket{\phi}$ as possible. We assume that $\eps$ is some constant strictly less than $1/2$ and will return to $\eps$-dependence at the end of this section.

\begin{theorem}\label{thm:ip_lb}
Suppose Alice and Bob share a $q$-qubit entangled state, can perform measurements on $k$ copies of their respective $d$-dimensional states $\ket{\psi}$ and $\ket{\phi}$ respectively, and engage in unbounded classical communication. If they are able to estimate $| \langle \psi | \phi \rangle |^2$ to accuracy $\varepsilon$ with high probability,~then $k= \Omega\left(\sqrt{d/q}\right)$.
\end{theorem}

\begin{theorem}\label{thm:ip_ub}
 Let $q=\Omega(\log d)$. Using $k=\Oh(\sqrt{d/q})$ copies of $\ket{\psi}$ and $\ket{\phi}$, there is a protocol that uses $\Theta(1)$ $q$-dimensional quantum messages, $\LOCC$, and returns $| \langle \psi | \phi \rangle |^2$ up to constant error with high probability
\end{theorem}

We remark that there is a constant-factor difference in the amount of entanglement bounds  in the upper and lower bound. This largely seems to stem from the sample complexity still being greater than $1$ even when $q=n$: say that Alice is  able to transmit her state exactly and Bob does the swap test. Even then, in order to compute the inner product $|\langle \psi|\phi\rangle|^2$ to error $0.1$, they need to repeat the swap test constantly many times. Our lower bound for this setting would imply a constant lower bound as well (matching the upper bound, but not ``exactly" since it would be a constant-factor off). Regardless, our lower bound can be understood as saying that entanglement does not help for inner product estimation unless the shared entanglement dimension scales with~$d$. 

\subsection{Lower bound in main theorem}
To obtain a lower bound we consider the decision problem constructed in~\cite{anshu2022distributed} (which is called the $\DIPE$, \emph{distributed inner product estimation} problem). Alice and Bob are promised to be in one of the following two scenarios. Their version restricts Alice and Bob to $\LOCC$ protcols. However, here we allow them to share some resource state $\sigma$ as well.
\begin{definition}[Distributed Inner Product Estimation, Decision Version]\label{def:DIPE}
\ 
Alice and Bob each have access to $k$ copies of the states $\ket{\phi}$ and $\ket{\psi}$ in $\mathbb{C}^d$ respectively and are asked to decide, using $\LOCC$ and perhaps a resource state $\sigma$, which of the following two scenarios they are in:
\begin{itemize}
    \item $\YES$: $\ket{\phi}=\ket{\psi}$ and they are Haar random
    \item $\NO$: $\ket{\phi},\ket{\psi}$ are Haar random.
\end{itemize}
\end{definition}

Note that the ability to transmit an arbitrary $q$-dimensional quantum state is equivalent, under $\LOCC$, to sharing a $q$-dimensional maximally entangled state (via quantum teleportation). Going forward, we let $\sigma=\frac{1}{q}\sum_{i,j=1}^q \ketbra{ii}{jj}$ be such a state. We also introduce a measure of entanglement called the \textit{robustness of entanglement}.
\begin{definition}[Robustness of entanglement~\cite{vidal1999robustness}]
   Any quantum state $\sigma\in \mathcal{B}(\mathcal{H}_A\otimes \mathcal{H}_B)$ can be decomposed as $\sigma = (1+s)\sigma^+ - s \sigma^-$,   where $\sigma^+$ and $\sigma^-$ are both separable states and $s\in \mathbb{R}_{\geq 0}$. The minimum value of $s$ over all such decompositions is called the \emph{robustness of entanglement}. We denote this minimum value by $E(\sigma)$.
\end{definition}
In particular, we will we will use the following lemma.
\begin{lemma}[\cite{vidal1999robustness}]\label{lem:RE}
    If $\sigma$ is a pure bipartite maximally entangled state of dimension $d$, then its robustness of entanglement is $E(\sigma) = q-1$.
\end{lemma}
In~\cite{anshu2022distributed} they show the following result:
\begin{theorem}[\cite{anshu2022distributed}]\label{thm:dipe_hard}
    Let $M$ be a separable measurement, then 
    \begin{align}
        \left| \E_{\phi,\psi} \left[M\cdot \Tr[\phi^{\otimes k}\otimes(\phi^{\otimes k} - \psi^{\otimes k})] \right] \right| \leq e^{k^2/d}-1\ .
    \end{align}
\end{theorem}
With this we now prove our main theorem lower bound. 

\begin{proof}[Proof of lower bound of Thm~\ref{thm:ip_lb}]   Suppose there exists a separable protocol $\mathcal{A}$ that uses $k$ copies of the states $\ket{\phi}$ and $\ket{\psi}$ and the resource state $\sigma$ and returns an estimate of $| \langle \psi | \phi \rangle |^2$ to accuracy~$\varepsilon < 1/2$ with probability $\geq 3/4$. By the  anti-concentration of the Haar measure, Alice and Bob will be able to use this protocol to solve problem~\ref{def:DIPE} with high probability. To see this, let $B([\phi], \eps')$ be the ball of radius $\eps' < \pi/2$ around $\phi$ in $\mathbb{CP}(d)$ with respect to the metric $d([\phi],[\psi]) = \arccos | \langle \psi | \phi \rangle | \in [0,\pi/2]$. Then, it is known~\cite{brannan2021alice} that, under the uniform distribution, the volume of $B([\phi], \eps')$ is $\sin^{2d-2}\eps'$. Alice and Bob use $\mathcal{A}$ to estimate $\vert \langle \phi \vert \psi\rangle \vert^2$ to precision $\varepsilon$ and output $\YES$ ($\NO$) if they measure the overlap to be at least $1-\varepsilon$ (at most $\varepsilon$). Clearly this succeeds with probability at least $3/4$ in the $\YES$ instance. In the $\NO$ instance, the ball $B([\phi], \arccos{\sqrt{\eps}})$ has measure $(1-\eps)^{2d-2}$. With probability at least $\left(1-(1-\eps)^{2d-2}\right)3/4$, they successfully output $\NO$. Thus, if Alice and Bob use the protocol $\mathcal{A}$, they will solve problem~\ref{def:DIPE} with probability at least $2/3$ as well (for $d$ large enough). This implies that there is a separable measurement $\{M,\id-M\}$ such that
    \begin{align}
    \label{eq:lowerboundonLBphipsisigma}
        1/3 \leq \left| \E_{\phi,\psi}\Tr\left[M (\phi^{\otimes k} \otimes (\phi^{\otimes k} - \psi^{\otimes k})\otimes \sigma)\right] \right|. 
    \end{align}
    We now show that the RHS of the inequality above can be upper bounded by $(1+2E(\sigma))(e^{k^2/d}-1)$. To see this split $\sigma$ into $\sigma = (1+E(\sigma))\sigma^+ - E(\sigma) \sigma^-$, where $\sigma^+$ and $\sigma^-$ are both separable states. Then, it follows that 
    \begin{align}
        \Tr[M(\phi^{\otimes k}\otimes(\phi^{\otimes k}- \psi^{\otimes k})\otimes \sigma)]  = &\ (1+E(\sigma))\Tr[M(\phi^{\otimes k}\otimes(\phi^{\otimes k} -\psi^{\otimes k})\otimes \sigma^+]\notag \\
        & + E(\sigma)\Tr[M(\phi^{\otimes k}\otimes(\psi^{\otimes k} -\phi^{\otimes k})\otimes \sigma^-]\ .
    \end{align}
    Note that the sign of $\sigma^-$ has been absorbed into the trace in the second term above. Now, because $\sigma^\pm$ are separable, each term above could be replaced by a separable measurement on $\phi^{\otimes k}\otimes(\phi^{\otimes k} - \psi^{\otimes k})])$ without any reference state. To see this, decompose $\sigma^+$ into $\sigma^+ = \sum_i p_i \rho_i^A \otimes \rho_i^B$. Using the spectral decompositions $\rho_i^A = \sum_k \lambda_{i,k} \ketbra{u_{i,k}}{u_{i,k}}$ and $\rho_i^B = \sum_k \nu_{i,k} \ketbra{v_{i,k}}{v_{i,k}}$, we arrive at
    \begin{align}
        &\Tr[M(\phi^{\otimes k} \otimes (\phi^{\otimes k}- \psi^{\otimes k}) \otimes \sigma^+)]\\
        & = \sum_i p_i \sum_{j,k}\lambda_{i,k}\nu_{i,j} \Tr[M(\phi^{\otimes k} \otimes \psi^{\otimes k} \otimes \ketbra{u_{i,j}}{u_{i,j}} \otimes \ketbra{v_{i,k}}{i,k})].
    \end{align}
    As a separable measurement, we have $M = \sum_t A_t \otimes B_t$. Now define a new measurement
    \begin{align}
        M' := \sum_t \sum_i p_i \left( \sum_j \lambda_{i,j} A_t^{i,j,j} \right) \otimes \left( \sum_k \nu_{i,k} B_t^{i,k,k}\right)\ ,
    \end{align}
    where $A_t = \sum_{j,j'} A_t^{i,j,j'}\otimes \ketbra{u_{i,j}}{u_{i,j'}}$ and $B_t = \sum_{k,k'} B_t^{i,k,k}\otimes \ketbra{v_{i,k}}{v_{i,k'}}$. As a sum of positive operators, this is positive. As a convex combination of operators majorized by $\id$, this is also majorized by $\id$. It then follows that $M'$ is a separable POVM. From Thm~\ref{thm:dipe_hard} it then follows that
    \begin{align}
        | \E_{\phi, \psi}\left[ \Tr[M (\phi^{\otimes k} \otimes (\phi^{\otimes k} - \psi^{\otimes k}) \otimes \sigma^+)] \right]| & = |\E_{\phi, \psi}\left[ \Tr[M' (\phi^{\otimes k} \otimes (\phi^{\otimes k} - \psi^{\otimes k})] \right]| \notag \\
        & \leq e^{k^2/d}-1\ .
    \end{align}
    The same proof shows that
    \begin{align}
        | \E_{\phi, \psi}\left[ \Tr[M (\phi^{\otimes k} \otimes (\psi^{\otimes k} - \phi^{\otimes k}) \otimes \sigma^-)] \right]| & \leq e^{k^2/d}-1\ .
    \end{align}
    Thus, we have that
    \begin{align}
        \left| \E_{\phi, \psi}\left[ \Tr[M (\phi^{\otimes k} \otimes (\phi^{\otimes k} - \psi^{\otimes k}) \otimes \sigma)] \right] \right| & \leq (1+2E(\sigma))(e^{k^2/d}-1)\ .
    \end{align}
    Combining the above along with our lower bound in Eq.\eqref{eq:lowerboundonLBphipsisigma} shows
     \begin{align}
        1/3 \leq \left| \E_{\phi,\psi}\Tr\left[M (\phi^{\otimes k} \otimes (\phi^{\otimes k} - \psi^{\otimes k})\otimes \sigma)\right] \right|
        \leq (1+2E(\sigma))(e^{k^2/d}-1)\ .
    \end{align}
    Rearranging, this implies that
    \begin{align}
        k = \Omega\left(\sqrt{\frac{d}{\log E(\sigma)}}\right).
    \end{align}
    So, if Alice and Bob share $q$ Bell pairs, then $E(\sigma) = 2^q-1$ and we get the desired lower~bound.
\end{proof}

\subsection{Upper bound in main theorem}
We now give a protocol which achieves the promised sample complexity of Thm~\ref{thm:ip_ub}. At a high level, the protocol works by randomly projecting $\ket{\phi}$ and $\ket{\psi}$ into small-dimensional subspaces. Since such a projection maintains the inner product between two stats with high probability, Alice can simply send Bob these smaller dimensional states and Bob can perform a swap test. The reader may notice that this assumes that Alice can transmit an arbitrary state in a $q$-dimensional subspace of a $d$-dimensional space. This, however, does not lead to any issues since Alice \textit{knows} the subspace the state lies in. Then, any maximally entangled state on a $q$-dimensional subspace can be converted to a maximally entangled state on the desired subspace via a unitary transformation. From there, Alice and Bob simply perform a teleportation protocol.

\begin{theorem}
    Let $q=\Omega(\log d)$. Using $k=\Oh(\sqrt{d/q})$ copies of $\ket{\psi}$ and $\ket{\phi}$, there is a protocol that uses $\Theta(1)$ $q$-dimensional quantum messages, $\LOCC$, and returns $| \langle \psi | \phi \rangle |^2$ up to constant error with high probability
\end{theorem}
\begin{proof}
    Let $\psi = \ketbra{\psi}{\psi}$ and $\phi = \ketbra{\phi}{\phi}$. Fix some unitary $U$. For convenience of notation we will assume that $q$ divides $d$. Divide the $d$-dimensional Hilbert space into $d/q$ subspaces each of dimension $q$. Let this decomposition be given by a collection of orthogonal projectors $\{P_i\}_{i=1}^{d/q}$. Alice and Bob perform the protocol in Figure~\ref{fig:ent_ip}.
   
    \begin{figure}[!ht]
    \begin{tcolorbox}
    \textbf{\emph{Input}}: Alice gets $\ket{\psi}^{\otimes k}$, Bob gets $\ket{\phi}^{\otimes k}$\\
    \textbf{\emph{Output}}: $\varepsilon$-approximation of $|\langle \psi|\phi\rangle|^2$
  \begin{enumerate}
        \item Using shared randomness, Alice and Bob sample a unitary $U$ from Haar~measure.
        \item Alice and Bob apply the POVM $\{U P_i U^\dagger \}_{i=1}^{\lceil d/q \rceil}$ on each of the $k$ copies and record which copies were projected onto which subspaces.
        \item Using a two-way classical communication they determine which copies lie now in the same subspaces. Alice sends Bob a constant number of such projected states which can be paired with a post-projection state of Bob's. Say the number of such pairs is~$m$.
        \item Bob performs a $\mathsf{SWAP}$ test between the post-projected pairs of states lying in the same subspaces. We say a $\mathsf{SWAP}$ test succeeds if Bob measures $\ket{0}$ on the ancilla register. Let the number of successes be $s$.
        \item Bob declares the inner product between $\ket{\psi}$ and $\ket{\phi}$ to be $2s/m-1$.
    \end{enumerate}
    \end{tcolorbox}
    \caption{Protocol to estimate inner product using shared entanglement.}
    \label{fig:ent_ip}
    \end{figure}
    
Our theorem is witnessed by repeating the protocol in Figure~\ref{fig:ent_ip} constantly many times, each time sending a $q$-dimensional message. To see the correctness, we will require the following~fact.
    \begin{fact}[{\cite[Fact~2]{sen2018quantum}}]\label{fact:haar_norm}
       Let $1\leq q\leq d$. Let $v\in \mathbb{C}^d$ be a unit vector and $P_i$ be an arbitrary projector onto a $q$-dimensional subspace. Let $U$ be drawn from the unitary Haar measure.~Then,
        \begin{align}
            \Pr_U\left[ \Vert P_i U  v \Vert_2 \not\in (1\pm \Delta) \sqrt{\frac{q}{d}} \right] \leq 4\exp\left\{ -\frac{\Delta^2 q}{16} \right\},
        \end{align}
    \end{fact}
    For a unitary $U$, let $\ket{\psi_i} = \frac{P_i U\ket{\psi}}{| P_i U\ket{\psi}|}$ and $\ket{\phi_i} = \frac{P_i U\ket{\phi}}{| P_i U\ket{\phi}|}$ be the normalized projections of $\ket{\psi}$ and $\ket{\phi}$ into the subspace given by $P_i$. We will now show that, with high probability over the choice of $U$, $| \langle \psi_i | \phi_i \rangle |^2 \approx | \langle \psi | \phi \rangle |^2$.

    \begin{lemma}
        Let $\Delta > 0$. Then, 
        \begin{align}
           \Pr_U \Big[\left| \ | \langle \phi_i' | \psi_i' \rangle |^2 - | \langle \phi | \psi \rangle |^2 \ \right| \leq 16\Delta\ \text{ for every }i \in [d/q]\Big]\geq 1-O\big(d/q\cdot \exp(-\Delta^2 q/16)\big).
        \end{align}
    \end{lemma}
    \begin{proof}
        Using Fact~\ref{fact:haar_norm}, with probability at least $1-16\exp\{-\frac{\Delta^2 q}{16}\}$, the following four conditions~hold
        \begin{align}
            \| P_i U \ket{\psi} \|_2,   \| P_i U \ket{\phi} \|_2 & \in (1\pm\Delta)\sqrt{\frac{q}{d}}\ ,\\
            \| P_i U (\ket{\psi} - \ket{\phi}) \|_2 & \in (1\pm\Delta)\sqrt{\frac{q}{d}}\| \ket{\psi} - \ket{\phi}\|_2\ ,\\
            \| P_i U (\ket{\psi} - i\ket{\phi}) \|_2 & \in (1\pm\Delta)\sqrt{\frac{q}{d}}\| \ket{\psi} - i\ket{\phi}\|_2\ .
    \end{align}
    We now follow steps similar to those of~\cite[Theorem 1]{sen2018quantum} to obtain
    \begin{align}
        \vert \langle \phi_i | \psi_i \rangle - \langle \psi | \phi \rangle \vert \leq 8\Delta\ .
    \end{align}
    Let $x, y \in \mathbb{C}$ be such that $| x- y | \leq \Delta$. Then,
    \begin{align}
        | | x |^2 - | y |^2 | & = | (| x | - | y |) (| x | + | y |) | \leq \Delta (| x | + | y |)\ .
    \end{align}
    In our case, this implies that
    \begin{align}
        | \langle \phi_i | \psi_i \rangle |^2 \in | \langle \psi | \phi \rangle |^2 \pm 16\Delta\ . 
    \end{align}
    Taking a union bound over all $d/q$ subspaces completes the proof.
    \end{proof}
Now, after receiving $m$ pairs of post-projected states, Bob then performs $m$ many $\swap$ tests. A $\swap$ test on the pair $\ket{\psi'_i}$ and $\ket{\phi'_i}$ succeeds with probability $\frac{1}{2}+\frac{1}{2}| \langle {\phi'_i} | {\psi'_i} \rangle |^2 \in \frac{1}{2} + \frac{1}{2}| \langle \phi | \psi \rangle | \pm 8\Delta$. The expected value of the sample average $s/m$ then satisfies
    \begin{align}
        \left | \E\left[ \frac{s}{m} \right] - \frac{1}{2}-\frac{1}{2}| \langle \psi | \phi \rangle |^2 \right| \leq 8\Delta,
    \end{align}
which implies
   $     2\E\left[ \frac{s}{m} \right]-1 \in | \langle \phi | \psi \rangle|^2 \pm 16\Delta,
    $ and by a Hoeffding bound, we get
    \begin{align}
        2\left| \frac{s}{m} - \E\left[ \frac{s}{m} \right] \right| \leq \delta\ ,
    \end{align}
    with probability at least $1-2\exp(-2m\delta^2)$. In this case, the total error of the estimator is at most $16\Delta + \delta$. For constant error, $m=O(1)$ suffices. 
    
     It remains to determine $k$ needs to be to obtain $m=O(1)$ projected pairs with high probability.
To this end, let $E_{a,b}$ be an indicator variable for if Alice's $a$th projection falls into the same subspace as Bob's $b$th projection. The expected number of collisions is given~by 
     \begin{align}
          \sum_{a,b}\E[E_{a,b}] = k^2 \sum_i \Tr[P_i U \psi U^\dagger] \Tr[P_i U \phi U^\dagger] & \geq k^2 (1-\Delta)^4\frac{q}{d}\ ,
     \end{align}
     where the inequality follows from $\| P_i U \ket{\psi}\|, \| P_i U \ket{\phi}\| \geq (1-\Delta)\sqrt{\frac{q}{d}}$. By a Hoeffding bound, it suffices to take $k=\Omega(\sqrt{d/q})$ to obtain a collision with high probability. As we only require a constant number of pairs, this can simply be repeated a constant number of times to obtain a constant number of collisions with high probability. Thus, the protocol succeeds for $k=\Omega(\sqrt{d/q})$.
  It remains to pick $\Delta$ and $\delta$ such that the claims go through: we require that $\frac{d}{q}\exp(-{\Delta^2 q}/{16}) = O(1)$, 
  which can be achieved by letting $\Delta=O(1)$ and~$q=\Omega(\log d)$.
    \end{proof}

\subsection{Bounds Depending on $\eps$}

\begin{theorem}\label{thm:ip_lb_eps}
Suppose Alice and Bob share a $q$-qubit entangled state, can perform measurements on $k$ copies of their respective $d$-dimensional states $\ket{\psi}$ and $\ket{\phi}$ respectively, and engage in unbounded classical communication. If they are able to estimate $| \langle \psi | \phi \rangle |^2$ to accuracy $\varepsilon$ with high probability,~then $k\geq \Omega(\max\{1/\varepsilon^2,\sqrt{d/q}/\varepsilon\})$.
\end{theorem}
\begin{proof}
    The lower bound of $\Omega(1/\eps^2)$ follows from lemma 13 of \cite{anshu2022distributed}. For the other lower bound, we use similar techniques to~\cite[Section~5.3]{anshu2022distributed}, wherein they showed how to use standard concentration techniques to reduce the decision version of DIPE (for which we have a $\Omega(\sqrt{d/q})$ lower bound) to prove a lower bound for estimating the inner product between $\ket{\phi}$ and $\ket{\psi}$ with a factor of $1/\varepsilon$.  In particular, they prove that if there exists a protocol that solves the $\varepsilon$-version of inner product estimation then it can be used to solve $\DIPE$ with high probability. Then, the lower bound on $\DIPE$ kicks in from their reduction.
\end{proof}

For an upper bound, we simply note that Alice and Bob can repeat the protocol in figure~\ref{fig:ent_ip} $\Theta(1/\eps^2)$ times to receive an $\eps$-approximation. However, $q$ must be sufficiently large such that the inner products are maintained with high probability to error $\Oh(\eps)$. A simple modification of the prior argument shows that $q=\Omega((\log d)/\eps^2)$ suffices.

\begin{theorem}
    Let $q=\Omega(\log d / \eps^2)$. Using $k=\Oh(\sqrt{d/q}/\eps^2)$ copies of $\ket{\psi}$ and $\ket{\phi}$, there is a protocol that uses $\Theta(1/\eps^2)$ $q$-dimensional quantum messages, $\LOCC$, and returns $| \langle \psi | \phi \rangle |^2$ up to error $\eps$ with high probability.
\end{theorem}
\begin{proof}
    To obtain an $\varepsilon$-approximation, $1/\varepsilon^2$ rounds of the $\swap$ test are necessary and sufficient by a Hoeffding bound. Similarly, it suffices for Alice and Bob to project onto random subspace $\Oh(\sqrt{d/q}/\varepsilon^2)$ times. We require that $\Delta$ as in the prior proof is on the order of $\Oh(\eps)$. Then, for the union bound argument to go through, we take $q=\Omega(\log d/ \eps^2)$.
\end{proof}

    \section{Generalized Distributed Inner Product Estimation}
    Now we turn to bilinear forms $f: \mathbb{C}^d \otimes \mathbb{C}^d \rightarrow \mathbb{C}$. Any such form can be expressed as $f(u,v) = u^\dagger M v$ for a matrix $M \in \mathbb{C}^{d \times d}$. Here we will assume that $M$ is Hermitian, which implies that $f(u,v) = \overline{f(v,u)}$.  Without loss of generality we will assume $M$ to be normalized such that $\Vert M \Vert = 1$. Due to the unphysical nature of global phases, $f(\ket{\psi}, \ket{\phi})$ is less meaningful than $| f(\ket{\psi}, \ket{\phi}) |^2$ and we concern ourselves entirely with this value instead. Of course, the special case of $M=\id$ corresponds to inner product estimation. Fixing such an $M$, let $P_\varepsilon$ be a projector which annihilates all eigenspaces of $M$ of eigenvalue of magnitude less than $\varepsilon/2$. We define $d_\varepsilon$ to be the dimension of the support of $P_\varepsilon$. Now define $M_\varepsilon = P_\varepsilon M P_\varepsilon$. The intuition behind this is that we can drop eigenspaces with small weights in estimating $\vert \langle \phi \vert M \vert \psi \rangle \vert^2$. The following lemma formalizes this.

    \begin{lemma} 
        Let $M$ be a Hermitian operator and $M_\varepsilon = P_\varepsilon M P_\varepsilon$, then 
        \begin{align}
            \left\vert \vert \langle \phi \vert M \vert \psi \rangle \vert^2 - \vert \langle \phi \vert M_\varepsilon \vert \psi \rangle \vert^2 \right\vert \leq \varepsilon/2\ .
        \end{align}
    \end{lemma}
    \begin{proof}
        \begin{align}
            | \Tr[M\psi M \phi] - \Tr[M_\varepsilon \psi M_\varepsilon \phi] | & = | \Tr[\psi (M\phi M - M_\varepsilon \phi M_\varepsilon)] | \\
            & = |\text{vec}(\psi)^\dagger (M\otimes M^\top - M_\varepsilon\otimes M_\varepsilon^\top)\text{vec}(\phi) | \\ 
            & \leq \Vert M\otimes M - M_\varepsilon\otimes M_\varepsilon \Vert\leq \varepsilon /2\ ,
        \end{align}
        where the second line uses the linear map $\text{vec}: \mathcal{B}(\mathcal{H}) \rightarrow \mathcal{H}\otimes \mathcal{H}$ defined by $\text{vec}(\ket{i}\bra{j}) = \ket{i}\otimes\ket{j}$ and the identity $\text{vec}(AXB) = (A\otimes B^\top) \text{vec}(X)$. The third line follows from the transpose map not changing the eigenvalues of a Hermitian operator.
    \end{proof}
 Using this lemma, if Alice and Bob can estimate $ \vert \langle \phi \vert M_\varepsilon \vert \psi \rangle \vert^2$ up to precision $\varepsilon/2$, that directly implies an $\varepsilon$-approximation to the quantity $| \langle \phi \vert M \vert \psi \rangle |^2$. For the remainder of this section, we will be primarily focused on upper and lower bounds for estimating $ \vert \langle \phi \vert M_\varepsilon \vert \psi \rangle \vert^2$.

With this notation, in this section our main result will be the following theorem, proving close-to-optimal upper and lower bounds for estimating bilinear forms on $\ket{\psi},\ket{\phi}$.

    \begin{theorem}\label{thm:blf_ub}
      For every $M$ with $\|M\|\leq 1$, with only $\LOCC$, it is sufficient to obtain 
      $$
      k=O(\max \{1/\eps^2, \| M_\eps \|_2/\eps\})
      $$ copies of $\ket{\psi},\ket{\phi}$  to estimate $| \langle \phi | M | \psi \rangle |^2$ to error $\varepsilon$ with high prob.~and it necessary to~obtain 
      $$
      k=\Omega(\max \{1/\eps^2, \| M_\eps \|_2/\sqrt{\eps}\})
      $$ copies to produce an $\varepsilon$-approximation to $| \langle \phi | M | \psi \rangle |^2$.
    \end{theorem}

    \subsection{Upper Bound}

        \begin{figure}[t]
        \begin{tcolorbox}
        \textbf{\emph{Input}}: An operator $M$. Alice gets $\ket{\psi}^{\otimes k}$, Bob gets copies of $\ket{\phi}^{\otimes k}$.\\
        \textbf{\emph{Output}}: An $\varepsilon$-approximation of $\langle \psi|M|\phi\rangle|^2$.
        \begin{enumerate}
        \item Alice and Bob each perform the two-outcome measurement $\{P_\eps, \id - P_\eps\}$  
 on each of the $k$ copies of their respective states  and obtain $s_a$ and $s_b$ copies of the states projected into one of the two subspace.   If $s_a=0$ or $s_b=0$ then Alice and Bob simply output $0$.
         \item Alice and Bob independently perform the standard POVM in the symmetric subspace on the support of $M_\eps$, obtaining the classical outcomes $\ket{u}$ and $\ket{v}$.

        \item Alice communicates $\ket{u}$ and $s_a$ to Bob who outputs
         \begin{align}
            w := \frac{(d_\varepsilon+s_a)(d_\varepsilon+s_b)}{k^2}| \langle u | M_\eps | v \rangle |^2 - \frac{\Tr[M_\varepsilon^2]}{k^2} \ .
        \end{align}
        \end{enumerate}
        \end{tcolorbox}
        \caption{Protocol to estimate $| \langle \phi \vert M_\varepsilon \vert \psi \rangle |^2$ using only $\LOCC$.}
        \label{fig:bilinear_protocol}
        \end{figure}
    
        In Figure~\ref{fig:bilinear_protocol} we outline a protocol that estimates $| \langle \phi \vert M_\varepsilon \vert \psi \rangle |^2$. Recall that the standard POVM on the symmetric subspace, def~\ref{def:standard_POVM}, has continuous aspects $\left\{ \binom{d+k-1}{k} \varphi^{\otimes k}\ \vert \ \ket{\varphi}\in\mathbb{C}^d \right\}$. This can be extended to $\vee^k W$ for arbitrary subspaces $W \subseteq \mathbb{C}^d$:
        \begin{align}
            \left\{ \binom{\dim W+k-1}{k} \varphi^{\otimes k}\ \vert \ \ket{\varphi}\in W \right\}\ .
        \end{align}
        The second step of the protocol of Figure~\ref{fig:bilinear_protocol} has Alice and Bob implementing this POVM for $\text{Im}P_\varepsilon \subseteq \mathbb{C}^d$. Since they may not have $k$ copies after the projection step, this is a POVM on $\vee^{s_a} \text{Im}P_\varepsilon$ and $\vee^{s_b} \text{Im}P_\varepsilon$, respectively.        First note  that if $P_\varepsilon \ket{\psi} =0$ or $P_\varepsilon \ket{\phi} = 0$ then they always output $0$, which must be a good estimate in this case. Thus, we assume that $\ket{\psi_\varepsilon} = P_\varepsilon \ket{\psi}/\Vert P_\varepsilon \ket{\psi}\Vert$ and $\ket{\phi_\varepsilon} = P_\varepsilon \ket{\phi}/\Vert P_\varepsilon \ket{\phi}\Vert$ both exist. The technical lemmas that we prove are the mean and variance of our estimators, whose proofs appear in the next section.
\begin{lemma}\label{lem:bilinear_expected}
    The expected value of the estimator $w$ given in Figure~\ref{fig:bilinear_protocol} is 
    \begin{align}
        \E[w] = | \langle \phi | M_\varepsilon | \psi \rangle |^2 + \frac{\Tr[M_\varepsilon^2 \phi_\varepsilon]}{k}\Tr[P_\varepsilon \psi] + \frac{\Tr[M_\varepsilon^2 \psi_\varepsilon]}{k}\Tr[P_\varepsilon \phi]\ .
    \end{align}
\end{lemma}
\begin{lemma}\label{lem:bilinear_var}
    The variance of the estimator of Figure~\ref{fig:bilinear_protocol} is upper bounded by
    \begin{align}
        \Var(w) = \Oh\left( \frac{1}{k} + \frac{\Vert M_\varepsilon \Vert_2^2}{k^2} + \frac{\Vert M_\varepsilon \Vert^4}{k^4} \right)\ .
    \end{align}
\end{lemma}

Our estimator is biased, but the bias is at most $2/k$. Using Lemmas~\ref{lem:bilinear_expected} and \ref{lem:bilinear_var} we can prove the upper bound of Theorem~\ref{thm:blf_ub}. Letting 
\begin{align}
    k = \Omega\left(\max\left\{ \frac{1}{\varepsilon^2}, \frac{\Vert M_\varepsilon \Vert_2}{\varepsilon} \right\} \right)\ ,
\end{align}
we ensure that the variance is upper bounded by $O(\varepsilon^2)$ and that the bias is $\Oh(\varepsilon^2)$. Thus, with high probability, Alice and Bob's estimate is within $\varepsilon/2$\footnote{ Using Chebyshev's inequality, the total error of the estimator upper bounded by $\frac{\varepsilon}{2}+\frac{2}{k}+\frac{\varepsilon}{c}$ with high probability. Choosing constants properly, we can assume that $k \geq \frac{c}{\varepsilon^2}$. For $c \geq 1+\sqrt{5}$ this implies that $\frac{\varepsilon}{2}+\frac{2}{k}+\frac{\varepsilon}{c} \leq \varepsilon$.} of $\vert \langle \phi \vert M_\varepsilon \vert \psi \rangle\vert^2$ and thus within $\varepsilon$ of $\vert \langle \phi \vert M \vert \psi \rangle\vert^2$.

\subsection{Mean and variance computation}
Observe that one can decompose Alice and Bob's measured states as~\cite{fumio2006semicircle}
        \begin{align}
            \ket{u} & = \alpha e^{i\theta}\ket{\psi_\varepsilon} + \sqrt{1-\alpha^2}\ket{\psi'}\ ,\\
             \ket{v} & = \beta e^{i\varphi}\ket{\phi_\varepsilon} + \sqrt{1-\beta^2}\ket{\phi'}\ ,
        \end{align}
        where $\ket{\psi'}$  and $\ket{\phi'}$ are orthogonal to $\ket{\psi_\varepsilon}$ and $\ket{\phi_\varepsilon}$ respectively and $\alpha,\beta >0$. Note that $\alpha$ ($\beta$) and $\ket{\psi'}$ ($\ket{\phi'}$) are independent.
For convenience, we will adopt the notation 
        \begin{align}
            P_{\psi'}  = \id_{d_\varepsilon} - \ketbra{\psi_\varepsilon}{\psi_\varepsilon}\ , & \quad P_{\phi'}  = \id_{d_\varepsilon} - \ketbra{\phi_\varepsilon}{\phi_\varepsilon}\\
            \swap_{\psi'} = P_{\psi'}^{\otimes 2}\swap  P_{\psi'}^{\otimes 2}\ , & \quad \swap_{\phi'} = P_{\phi'}^{\otimes 2}\swap  P_{\phi'}^{\otimes 2}\ .
        \end{align}
We will often implicitly use the fact that
        \begin{align}
            \swap_{\psi'}  =  P_{\psi'}^{\otimes 2}\swap  P_{\psi'}^{\otimes 2} =  P_{\psi'}^{\otimes 2}\swap =  \swap  P_{\psi'}^{\otimes 2}\ .
        \end{align}
 A similar equivalence holds for $\swap_{\phi'}$. Also define $f_\varepsilon = | \langle \phi_\varepsilon | M_\varepsilon | \psi_\varepsilon \rangle |^2$. Lastly, we will need the following facts about these random variables.
    
\begin{fact}[{\cite[Lemma 3]{anshu2022distributed}}]\label{fact:std_povm} For $\ket{\psi'},\ket{\phi'}$ as defined above, we have
            \begin{align}
                \E[\ketbra{\psi'}{\psi'}]  = \frac{P_{\psi'}}{d_\varepsilon-1}\ ,& \quad \E[\ketbra{\phi'}{\phi'}]  = \frac{P_{\phi'}}{d_\varepsilon-1}\ .\\
                \E[\psi'^{\otimes 2}] = \frac{1}{d_\varepsilon(d_\varepsilon-1)}\left[ P_{\psi'}^{\otimes 2} + \swap_{\psi'}\right]\ , & \quad
                \E[\phi'^{\otimes 2}] = \frac{1}{d_\varepsilon(d_\varepsilon-1)}\left[ P_{\phi'}^{\otimes 2} + \swap_{\phi'}\right]\\
                \E[\alpha^2] = \frac{s_a+1}{d_\varepsilon+s_a}\ , & \quad \E[\beta^2] =\frac{s_b+1}{d_\varepsilon+s_b}\\
                \E[\alpha^4] = \frac{(s_a+2)(s_a+1)}{(d_\varepsilon+s_a+1)(d_\varepsilon+s_a)}\ , & \quad \E[(1-\alpha^2)^2] = \frac{d_\varepsilon(d_\varepsilon - 1)}{(d_\varepsilon + s_a + 1)(d_\varepsilon + s_a)}\\
                \E[\beta^4] = \frac{(s_b+2)(s_b+1)}{(d_\varepsilon+s_b+1)(d_\varepsilon+s_b)}\ , & \quad \E[(1-\beta^2)^2] = \frac{d_\varepsilon(d_\varepsilon - 1)}{(d_\varepsilon + s_b + 1)(d_\varepsilon+s_b)}
            \end{align}
\end{fact}
Importantly, note that the expectation values for $\ket{\psi'}$ and $\ket{\phi'}$ do not depend on $s_a$ or $s_b$.

\subsubsection{Expected value}
In this section we will prove Lemma~\ref{lem:bilinear_expected}, i.e., the expected value of the estimator $w$ given in Figure~\ref{fig:bilinear_protocol}~is 
    \begin{align}
        \E[w] = | \langle \phi | M_\varepsilon | \psi \rangle |^2 + \frac{\Tr[M_\varepsilon^2 \phi_\varepsilon]}{k}\Tr[P_\varepsilon \psi] + \frac{\Tr[M_\varepsilon^2 \psi_\varepsilon]}{k}\Tr[P_\varepsilon \phi]\ .
    \end{align}
\begin{proof}

Using the tower property of expectation, we have that
\begin{align}
    \E[w] & = \E\big[\E[w | s_a, s_b]\big]\ ,
\end{align}
where the first expectation is over the measurement outcomes in Step 1 of the algorithm and the second is in the output of the standard POVM measurement. Now  if we expand $| \langle u | M_\varepsilon \vert v \rangle |^2$ using the decompositions of $\ket{u}$ and $\ket{v}$ above we obtain $16$ terms. However, all of the terms that are degree one in $\ket{\psi'}$ or $\ket{\phi'}$ are zero in expectation (by symmetry of the Haar measure), and we are left with $4$ terms:
\begin{align}
    &\E\big[| \langle u | M_\varepsilon | v \rangle |^2 \big\vert  s_a,s_b\big] \\
    &=  \E\big[\alpha^2 \beta^2 \vert s_a, s_b\big] \cdot f_\varepsilon + \E\big[(1-\alpha^2)\beta^2 \vert s_a, s_b\big]\cdot \E\big[| \langle \psi' | M_\varepsilon | \phi_\varepsilon \rangle |^2\big] \notag\\
                & + \E\big[\alpha^2(1-\beta^2) \vert s_a, s_b\big]\cdot\E\big[| \langle \psi_\varepsilon | M_\varepsilon | \phi' \rangle |^2\big] \notag\\
                & + \E\big[(1-\alpha^2)(1-\beta^2) \vert s_a, s_b\big]\cdot\E\big[| \langle \psi' | M_\varepsilon | \phi' \rangle |^2 \big], \notag
            \end{align}
where we should think of $\alpha,\beta$ above as random variables (depending on the randomness of the output $\ket{u}$).
        We now evaluate each term on the RHS as follows.
            \begin{align}
                \E[| \langle \psi' | M_\varepsilon | \phi_\varepsilon \rangle |^2] & = \E[ \Tr[M_\varepsilon \psi' M_\varepsilon \phi_\varepsilon] ] = \frac{1}{d_\varepsilon-1}\bra{\phi_\varepsilon} M_\varepsilon  P_{\phi'} M_\varepsilon \ket{\phi_\varepsilon} = \frac{\Tr[M_\varepsilon^2 \phi_\varepsilon] - f_\varepsilon}{d_\varepsilon-1}\ .
            \end{align}
            Similar steps yield that
            \begin{align}
                \E[| \langle \psi_\varepsilon | M_\varepsilon | \phi'\rangle |^2] & =  \frac{\Tr[M_\varepsilon^2 \psi_\varepsilon] - f_\varepsilon}{d_\varepsilon-1}\ .
            \end{align}
            Lastly,
            \begin{align}
                \E[| \langle \psi' | M_\varepsilon | \phi' \rangle |^2] & = \frac{1}{(d_\varepsilon-1)^2}\Tr\left[  M_\varepsilon P_{\psi'} M_\varepsilon P_{\phi'} \right]\\
                & = \frac{\Tr[M_\varepsilon^2] - \Tr[M_\varepsilon^2 \psi_\varepsilon] - \Tr[M_\varepsilon^2 \psi_\varepsilon] + f_\varepsilon}{(d_\varepsilon-1)^2}\ .
            \end{align}
            Putting everything above in along with  Fact~\ref{fact:std_povm}, we have that
            \begin{align}
                \E[| \langle u | M_\varepsilon | v \rangle |^2 \vert s_a, s_b] & = \frac{s_as_b}{(d_\varepsilon+s_a)(d_\varepsilon+s_b)}f_\varepsilon + \frac{\Tr[M_\varepsilon^2] + s_a\Tr[M_\varepsilon^2 \psi_\varepsilon] + s_b\Tr[M_\varepsilon^2 \phi_\varepsilon]}{(d_\varepsilon+s_a)(d_\varepsilon+s_b)}\\
                \E[w | s_a, s_b] & = \frac{s_a s_b}{k^2} f_\varepsilon + \frac{s_a\Tr[M_\varepsilon^2 \phi_\varepsilon] + s_b\Tr[M_\varepsilon^2 \psi_\varepsilon]}{k^2}\ .
            \end{align}
            Let $P$ and $Q$ be discrete distributions over $[k]$ such that $P(i)$ (resp.~$Q(i)$) is the probability that Alice (resp.~Bob) observes $i$ successful projections onto $P_\varepsilon$. Note $P$ and $Q$ are binomial distributions with parameters $p=\Tr[P_\varepsilon \psi]$ and $q=\Tr[P_\varepsilon \phi]$ respectively. Now we can compute $\E[w]$ as follows:            \begin{align}
                \E[w] & =  \sum_{s_a,s_b=1}^k P(s_a)Q(s_b)\left( \frac{s_a s_b}{k^2} f_\varepsilon + \frac{s_a\Tr[M_\varepsilon^2 \phi_\varepsilon] + s_b\Tr[M_\varepsilon^2 \phi_\varepsilon]}{k^2} \right)\\
                & = f_\varepsilon \left( \sum_{s_a=1}^k P(s_a) \frac{s_a}{k} \right) \left( \sum_{s_b=1}^k Q(s_b) \frac{s_b}{k} \right)\\
                & + \frac{\Tr[M_\varepsilon^2 \phi_\varepsilon]}{k} \sum_{s_a=1}^k P(s_a) \frac{s_a}{k} + \frac{\Tr[M_\varepsilon^2 \psi_\varepsilon]}{k} \sum_{s_b=1}^k Q(s_a) \frac{s_b}{k}\notag \\
                & = f_\eps \Tr[P_\varepsilon \phi] \Tr[P_\varepsilon \psi]+ \frac{\Tr[M_\varepsilon^2 \phi_\varepsilon]}{k}\Tr[P_\varepsilon \psi] + \frac{\Tr[M_\varepsilon^2 \psi_\varepsilon]}{k}\Tr[P_\varepsilon \phi]\\
                & = | \langle \phi | M_\varepsilon | \psi \rangle |^2 + \frac{\Tr[M_\varepsilon^2 \phi_\varepsilon]}{k}\Tr[P_\varepsilon \psi] + \frac{\Tr[M_\varepsilon^2 \psi_\varepsilon]}{k}\Tr[P_\varepsilon \phi]\ ,
            \end{align}
            where the last line follows from $\E\left[ \frac{s_a}{k} \right] = p$ and $\E\left[ \frac{s_b}{k} \right] = q$. This proves the claim bound.  
\end{proof}
\subsubsection{Bounding the variance}
Next we will prove our variance bound.  For ease of notation that comes with renormalization, we will let for the rest of this section. 
        \begin{align}
            D := (d_\varepsilon + s_a+1)(d_\varepsilon + s_a)(d_\varepsilon + s_b+1)(d_\varepsilon + s_b)\ .
        \end{align}
In this section we prove Lemma~\ref{lem:bilinear_var}, i.e., the variance of the estimator of Figure~\ref{fig:bilinear_protocol} is at most
    \begin{align}
        \Var(w) = \Oh\left( \frac{1}{k} + \frac{\Vert M_\varepsilon \Vert_2^2}{k^2} + \frac{\Vert M_\varepsilon \Vert^4}{k^4} \right)\ .
    \end{align}
\begin{proof}
Assuming the first inequality below, we prove the lemma. Deriving this inequality requires a sequence of lemmas which we prove in the following pages in Lemma~\ref{lem:univariate},~\ref{lemma:bivariate},~\ref{lem:quartic}:
            \begin{align}
                \E\big[ | \langle u | M_\varepsilon | v \rangle |^4 | s_a, s_b\big] \leq &  \frac{1}{D}(2(s_a+1)(s_b+1) + (s_a+2)(s_a+1)(s_b+2)(s_b+1)f_\varepsilon^2\\
                & + 2(s_a+2)(s_b+1) + 2(s_b+2)(s_b+1) + 2(\Tr[M_\varepsilon^2]+1)^2 +12 \\
                & + (s_a+1)(s_b+1) + (s_a+s_b+2)(\Tr[M_\varepsilon^2]+9)\\
                & + (s_a+1)(s_a+2)(s_b+1) + (s_a+1)(s_b+1)(s_b+2) \\
                & + (s_a+1)(s_b+1)\Tr[M_\varepsilon^2])\ .
            \end{align}
It follows then that
\[
\Var(w) = \frac{(d_\varepsilon+s_a)^2(d_\varepsilon+s_b)^2}{k^4}\Var(| \langle u | M_\varepsilon | v \rangle |^4).
\]
Then,
    \begin{align}
        \frac{1}{D} \cdot \frac{(d_\varepsilon+s_a)^2(d_\varepsilon+s_b)^2}{k^4} \leq \frac{1}{k^4}.
    \end{align}
    Since $s_a, s_b \leq k$, it follows that ${s_a^is_b^j}\cdot {k^{-4}} = \Oh\left( k^{i+j-4} \right)$ and further that
    \begin{align*}
        \Var(w) \leq & \Oh\left( \frac{1}{k} + \frac{\Vert M_\varepsilon \Vert_2^2}{k^2} + \frac{\Vert M_\varepsilon \Vert_2^4}{k^4} \right) \\
        & + \sum_{s_a, s_b=0}^k P(s_a)Q(s_a)\frac{s_a^2s_b^2}{k^4}f_\varepsilon - \left( \sum_{s_a,s_b=0}^k P(s_a) Q(s_a) \frac{s_as_b}{k^2} f_\varepsilon \right)^2\notag \ .
    \end{align*}
    Observe that the last two terms in the expression are the variance of the random variable $f_\varepsilon\cdot{s_a s_b}\cdot {k^{-2}}$. Since $s_a$ and $s_b$ are binomially distributed, this evaluates to exactly
    \begin{align}
        \Var\left( f_\varepsilon \frac{s_a s_b}{k^2} \right) & = \frac{f_\varepsilon^2}{k^4}\left( \E[s_a^2]\E[s_b^2]-\E[s_a]^2\E[s_b]^2 \right)\\
        & = \frac{f_\varepsilon^2}{k^4}\left( (kp(1-p)+k^2p^2)(kq(1-q)+k^2q^2)-k^4p^2q^2 \right)=O(1/k)
    \end{align}
\end{proof}

Over the next few pages we will go over tedious calculations to prove Lemma~\ref{lem:bilinear_var}. As we mentioned earlier, it suffices to compute
$$
\frac{(d_\varepsilon+s_a)(d_\varepsilon+s_b)}{k^2} \Var(| \langle u | M_\varepsilon | v \rangle |^2).
$$
For a fixed $s_a$ and $s_b$, define
the notation
            \begin{align}
                q = \alpha\beta\cdot  \langle \phi_\varepsilon | M_\varepsilon | \psi_\varepsilon \rangle\ , \quad g = \sqrt{1-\alpha^2}\beta\cdot  \langle \phi' | M_\varepsilon | \psi_\varepsilon \rangle\\
                h = \alpha\sqrt{1-\beta^2}\cdot  \langle \phi_\varepsilon | M_\varepsilon | \psi' \rangle\ , \quad \ell = \sqrt{1-\alpha^2}\sqrt{1-\beta^2} \cdot \langle \phi' | M_\varepsilon | \psi' \rangle\ .
            \end{align}.
            We now have that
     \begin{align}
          \E\big[| \langle u | M_\varepsilon | v \rangle |^4 \vert s_a, s_b\big]  = & \E[ | q|^2 + | g |^4 + | h |^4 + | \ell |^4 + \\
                & 4(| g |^2 | h |^2 + | \ell |^2 | g |^2 + | \ell |^2 | h |^2 + | q|^2 | g|^2 + | q|^2| h |^2 + | q |^2 | \ell |^2 +\\
                &q\cdot \ell\cdot g^*\cdot h^* + g\cdot h\cdot q^*\cdot \ell^*) ]\ ,
            \end{align}
           We upper bound all these terms below (the first line of univariate terms in Lemma~\ref{lem:univariate}, the second line of quadratic terms in Lemma~\ref{lemma:bivariate}, the third line of degree-$4$ terms in Lemma~\ref{lem:quartic}). For notational simplicity, we will drop the conditioning on $s_a$ and $s_b$ in expectation values in the next few lemmas.

            \begin{lemma}[Bounds on the univariate terms]
            \label{lem:univariate}
            We bound the $q,g,h,\ell$ term defined above as follows.
                \begin{align}
                    \E[| q |^4] & = \frac{(s_a+2)(s_a+1)(s_b+2)(s_b+1)}{D}f_\varepsilon^2\\
                     \E[| g |^4] &\leq \frac{2(s_b+2)(s_b+1)}{D}\\
                      \E[| h |^4 ] &\leq \frac{2(s_a+2)(s_a+1)}{D}\\
                      \E[| \ell |^4] &\leq \frac{2(\Tr[M_\varepsilon^2]+1)^2 + 12}{D}\ .
                \end{align}
            \end{lemma}
            \begin{proof}
           \textbf{\emph{Bound on term 1.}} As $\langle \phi_\varepsilon | M_\varepsilon | \psi_\varepsilon \rangle$ is a fixed, using Fact~\ref{fact:std_povm} we have that
           $$
           \E[\alpha^4\beta^4] = \frac{(s_a+2)(s_a+1)(s_b+2)(s_b+1)}{D},
           $$
           where we used that  $\alpha,\beta$ are independent.

        \textbf{\emph{Bound on term 2.}} Again using Fact~\ref{fact:std_povm}, we have 
          \begin{align}
          \E[| g |^4] & = \E[(1-\alpha^2)^2]\cdot \E[\beta^4]\cdot \E[| \langle \phi' | M_\varepsilon| \psi_\varepsilon \rangle |^4] \\
          & = \frac{(s_b+2)(s_b+1)d_\varepsilon(d_\varepsilon-1)}{D}\E[| \langle \phi' | M_\varepsilon| \psi_\varepsilon \rangle |^4]\ .
        \end{align}
          Now we bound the expectation term as
          \begin{align}
          \E[| \langle \phi' |M_\varepsilon| \psi_\varepsilon \rangle |^4] & = \E[\Tr[(\phi' M_\varepsilon \psi_\varepsilon M_\varepsilon)^{\otimes 2}]] \\
          & = \frac{1}{d_\varepsilon(d_\varepsilon - 1)}\Tr[(P_{\phi'}^{\otimes 2}+\swap_{\phi'})M_\varepsilon^{\otimes 2}\psi_\varepsilon^{\otimes 2} M_\varepsilon^{\otimes 2}]\ .
          \end{align}
                Direct evaluation yields that each term in the trace above is the same, i.e.,
                \begin{align}
                \label{eq:trpphiMepspsiepsMeps}
                    \Tr[P_{\phi'}^{\otimes 2} M_\varepsilon^{\otimes 2}\psi_\varepsilon^{\otimes 2} M_\varepsilon^{\otimes 2}] & = (\Tr[M_\varepsilon^2 \psi_\varepsilon]-f_\varepsilon)^2\ ,\\
                    \Tr[\swap_{\phi'}M_\varepsilon^{\otimes 2}\psi_\varepsilon^{\otimes 2} M_\varepsilon^{\otimes 2}] & = (\Tr[M_\varepsilon^2 \psi_\varepsilon]-f_\varepsilon)^2\ .
                \end{align}
                Furthermore, one can upper bound both the terms as $(\Tr[M_\varepsilon \psi_\varepsilon] - f_\varepsilon)^2 \leq 1$ since these quantities are in $[0,1]$. Combining the expressions derived thus far completes the proof.
                
                \textbf{\emph{Bound on term 3.}}   This follows from steps similar to those of term 2 and we omit the details.
                
\textbf{ \emph{Bound on  term 4.} }     Using Fact~\ref{fact:std_povm}, we have that
                \begin{align}
                    \E[(1-\alpha^2)^2(1-\beta^2)^2] = \frac{d_\varepsilon^2(d_\varepsilon-1)^2}{D}\ .
                \end{align}
                 Again, using Fact~\ref{fact:std_povm}, we obtain
                \begin{align}
                    \E[| \langle \phi' | M_\varepsilon | \psi' \rangle |^4] & = \frac{1}{d_\varepsilon^2(d_\varepsilon-1)^2}\Tr[(P_{\psi'}^{\otimes 2}+\swap_{\psi'})M_\varepsilon^{\otimes 2} (P_{\phi'}^{\otimes 2}+\swap_{\phi'})M_\varepsilon^{\otimes 2}]\ .
                \end{align}
                We split this computation into four parts. The first being
                \begin{align}
               \Circled{1} \equiv \Tr[P_{\psi'}^{\otimes 2} M_\varepsilon^{\otimes 2} P_{\phi'}^{\otimes 2} M_\varepsilon^{\otimes 2}] & = \Tr[P_{\psi'}^{\otimes 2} M_\varepsilon^{\otimes 2}(\id_{d_\varepsilon}^{\otimes 2} - \id_{d_\varepsilon} \otimes \phi_\varepsilon - \phi_\varepsilon \otimes \id_{d_\varepsilon} + \phi_\varepsilon^{\otimes 2})M_\varepsilon^2]\ .
                    \end{align}
                 We bound each of these individually as follows: using Eq~\eqref{eq:trpphiMepspsiepsMeps}, we have that
                    \begin{align}
                        \Tr[P_{\psi'}^{\otimes 2} M_\varepsilon^{\otimes 2}\phi_\varepsilon^{\otimes 2} M_\varepsilon^{\otimes 2}] & = (\Tr[M_\varepsilon^2 \phi_\varepsilon]-f_\varepsilon)^2\ .
                    \end{align}It is not hard to see that the first is 
                    \begin{align}
                        \Tr[(P_{\psi'}M_\varepsilon^2)^{\otimes 2}] & = \Tr[P_{\psi'}M_\varepsilon^2]^2 = (\Tr[M_\varepsilon^2]-\Tr[M_\varepsilon^2\psi_\varepsilon])^2\ .
                    \end{align}
                    The remaining two terms are equal and brute force evaluation yields
                    \begin{align}
                        \Tr[P_{\psi'}^{\otimes 2}M_\varepsilon^{\otimes 2}P_{\phi'}^{\otimes 2}M_\varepsilon^{\otimes 2}] & = \Tr[M_\varepsilon^2(\id_{d_\varepsilon} - \psi_\varepsilon - \phi_\varepsilon)]^2 + f_\varepsilon(f_\varepsilon +2 \Tr[M_\varepsilon^2 (\id_{d_\varepsilon} - \psi_\varepsilon - \phi_\varepsilon)])\ .
                    \end{align}
                Simplifying, we have that
                \begin{align}
                    \Circled{1} & = \Tr[M_\varepsilon^2(\id_{d_\varepsilon} - \psi_\varepsilon - \phi_\varepsilon)]^2 + f_\varepsilon(f_\varepsilon+2\Tr[M_\varepsilon^2(\id_{d_\varepsilon} - \psi_\varepsilon - \phi_\varepsilon)])\ .
                \end{align}

            The second is
            \begin{align}
                \Circled{2} & \equiv \Tr[P_{\psi'}^{\otimes 2}M_\varepsilon^{\otimes 2} \swap_{\phi'}M_\varepsilon^{\otimes 2}] \\
                & = \Tr[(\id_{d_\varepsilon}^{\otimes 2} - \id_{d_\varepsilon} \otimes \psi_\varepsilon - \psi_\varepsilon \otimes \id_{d_\varepsilon} + \psi_\varepsilon^{\otimes 2})M_\varepsilon^{\otimes 2}\swap_{\phi'}M_\varepsilon^{\otimes 2}]\ .
            \end{align}
            As some of these terms will appear in later calculations, we document them explicitly here:
            \begin{align}
            \Tr[\swap_{\phi'}(M_\varepsilon^2)^{\otimes 2}]  = \Tr[\swap (P_{\phi'}M_\varepsilon^2P_{\phi'})^{\otimes2}] &= \Tr[P_{\phi'} M_\varepsilon^2 P_{\phi'} M_{\varepsilon}^2]\\&= \Tr[M_\varepsilon^4] -2\Tr[M_\varepsilon^4\phi_\varepsilon] + \Tr[M_\varepsilon^2 \phi_\varepsilon]^2\  .
                    \end{align}
                    Next, we have
                    \begin{align}
                        \Tr[(\id_{d_\varepsilon}\otimes \psi_{\varepsilon})M_\varepsilon^{\otimes 2}\swap_{\phi'}M_\varepsilon^{\otimes 2}] = &  \Tr[(\id_{d_\varepsilon} - \phi_{\varepsilon})M_\varepsilon^2(\id_{d_\varepsilon} - \phi_{\varepsilon}) M_\varepsilon \psi_\varepsilon M_\varepsilon]\\
                         = & \Tr[M_\varepsilon^4 \psi_\varepsilon] + \Tr[M_\varepsilon^2 \phi_\varepsilon]f_\varepsilon\notag \\
                         & - 2\text{Re}(\langle \phi_\varepsilon | M_\varepsilon^3 | \psi_\varepsilon \rangle \langle \psi_\varepsilon | M_\varepsilon | \phi_\varepsilon \rangle)\ .
                    \end{align}
                    The third term is the same as the second. Lastly, the fourth term follows from the previous case above, and is $(\Tr[M_\varepsilon^2 \psi_\varepsilon] - f_\varepsilon)^2$. Combining, we have that
                    \begin{align}
                        \Circled{2} = & \Tr[M_\varepsilon^4(\id_{d_\varepsilon} - 2\psi_\varepsilon - 2\phi_\varepsilon)]+\Tr[M_\varepsilon^2 \psi_\varepsilon](\Tr[M_\varepsilon^2 \psi_\varepsilon]-2f_\varepsilon)+\Tr[M_\varepsilon^2 \phi_\varepsilon](\Tr[M_\varepsilon^2 \phi_\varepsilon]-2f_\varepsilon)\notag\\
                        & + 2\text{Re}(\langle \phi_\varepsilon | M_\varepsilon^3 | \psi_\varepsilon \rangle \langle \psi_\varepsilon | M_\varepsilon | \phi_\varepsilon \rangle)\ .
                    \end{align}
                    Similar steps yield that 
                    \begin{align}
                        \Circled{3} \equiv \Tr[\swap_{\psi'}^{\otimes 2}M_\varepsilon^{\otimes 2} P_{\phi'}^{\otimes 2} M_\varepsilon^{\otimes 2}] = \Circled{2}\ .
                    \end{align}
                    Lastly, we have
                    \begin{align}
                        \Circled{4} = & \Tr[\swap_{\psi'}M_\varepsilon^{\otimes 2}\swap_{\phi'} M_\varepsilon^{\otimes 2}] = \Tr[P_{\psi'}^{\otimes 2}M_\varepsilon^{\otimes 2}P_{\phi'}^{\otimes 2}M_\varepsilon^{\otimes 2}] = \Circled{1}\ .
                    \end{align}

                    We now simply drop all non-positive terms and arrive at
                    \begin{align}
                        \Circled{1}+\Circled{2} + \Circled{3} + \Circled{4} \leq 2\Tr[M_\varepsilon^2](\Tr[M_\varepsilon^2] + 2f_\varepsilon) + 2\Tr[M_\varepsilon^4] + 2f_\varepsilon^2 + 12\ .
                    \end{align}
                    Using $\Tr[M_\varepsilon^4] \leq \Tr[M_\varepsilon^2]$, we have 
                    \begin{align}
                        \Circled{1}+\Circled{2} + \Circled{3} + \Circled{4} \leq 2(\Tr[M_\varepsilon^2]+1)^2 + 12 \ .
                    \end{align}

            \end{proof}

             \begin{lemma}[Bounds on bivariate terms]
             \label{lemma:bivariate}
                \begin{align}
                    \E[| g|^2 | h |^2] &\leq \frac{(s_a+1)(s_b+1)f_\varepsilon^2}{D}\\
                    \E[| \ell|^2 | h|^2] &\leq \frac{(s_a+1)}{D}(\Tr[M_\varepsilon^2] + 9)\\
                       \E[| \ell|^2 | g|^2] &\leq \frac{(s_b+1)}{D}(\Tr[M_\varepsilon^2] + 9)\\
                       \E[| q |^2 | g |^2] &\leq\frac{(s_b+2)(s_b+1)(s_a+1)}{D}\\
                        \E[| q |^2 | h |^2] &\leq\frac{(s_a+2)(s_a+1)(s_b+1)}{D}\\
                        \E[ | q |^2 | \ell |^2] &\leq  \frac{(s_a+1)(s_b+1)}{D} \Tr[M_\varepsilon^2]\ .
                \end{align}
            \end{lemma}
            \begin{proof}
               \textbf{\emph{Bound on term 1.}} \begin{align}
                    \E[| g|^2 | h |^2] & = \E[\alpha^2(1-\alpha^2)]\E[\beta^2(1-\beta^2)] \E[| \langle \phi' | M_\varepsilon | \psi_\varepsilon \rangle |^2 \langle \phi_\varepsilon | M_\varepsilon | \psi'\rangle | |^2]\ .
                \end{align}
                First, we note that
                \begin{align}
                    \E[\alpha^2(1-\alpha^2)] & = \E[\alpha^2] - \E[\alpha^4] = \frac{(d_\varepsilon - 1)(s_a+1)}{(d_\varepsilon + s_a + 1)(d_\varepsilon + s_a)}\ .
                \end{align}
                Similarly,
                \begin{align}
                    \E[\beta^2(1-\beta^2)] & = \frac{(d_\varepsilon - 1)(s_b+1)}{(d_\varepsilon + s_b + 1)(d_\varepsilon + s_b)}\ .
                \end{align}
                We are left with computing 
                \begin{align}
                    \E[| \langle \phi' | M_\varepsilon | \psi_\varepsilon \rangle |^2 \langle \phi_\varepsilon | M_\varepsilon | \psi'\rangle |^2] & = \E[\Tr[\phi' M_\varepsilon \psi_\varepsilon M_\varepsilon]\Tr[\phi_\varepsilon M_\varepsilon \psi' M_\varepsilon]]\\
                    & = \E[[\Tr[\phi' M_\varepsilon \psi_\varepsilon M_\varepsilon]] \E[\Tr[\phi_\varepsilon M_\varepsilon \psi' M_\varepsilon]]\\
                    & = \frac{1}{(d_\varepsilon-1)^2}(\Tr[M_\varepsilon^2 \psi_\varepsilon]-f_\varepsilon)(\Tr[M_\varepsilon^2 \phi_\varepsilon]-f_\varepsilon)\ ,
                \end{align}
                where the second equality follows from independence in the expectation value.
            
                \textbf{\emph{Bound on term 2.}}   \begin{align}
                    \E[| \ell |^2 | h |^2] & = \E[\alpha^2(1-\alpha^2)] \cdot \E[(1-\beta^2)^2]\cdot \E[\Tr[(\phi_\varepsilon \otimes \phi') M_\varepsilon^{\otimes 2}(\psi')^{\otimes 2} M_\varepsilon^{\otimes 2}]]\\
                    & = \frac{(s_a+1)d_\varepsilon (d_\varepsilon - 1)^2}{D}\E[\Tr[(\phi_\varepsilon \otimes \phi') M_\varepsilon^{\otimes 2}(\psi')^{\otimes 2} M_\varepsilon^{\otimes 2}]]\ .
                \end{align}
                Then, we compute
                \begin{align}
                    \E[\Tr[(\phi_\varepsilon \otimes \phi') M_\varepsilon^{\otimes 2}(\psi')^{\otimes 2} M_\varepsilon^{\otimes 2}]] & = \frac{1}{d_\varepsilon (d_\varepsilon-1)^2}\Tr[(\phi_\varepsilon \otimes P_{\phi'}) M_\varepsilon^{\otimes 2} \swap_{\psi'} M_\varepsilon^{\otimes 2}]\\
                    & = \frac{1}{d_\varepsilon (d_\varepsilon-1)^2} \Tr[P_{\psi'}M_\varepsilon \phi_\varepsilon M_\varepsilon P_{\psi'} M_\varepsilon P_{\phi'} M_\varepsilon]\ .
                \end{align}
                We will omit the remaining calculations and state that this evaluates to 
                \begin{align}
                    d_\varepsilon (d_\varepsilon-1)^2\E[\Tr[(\phi_\varepsilon \otimes \phi') M_\varepsilon^{\otimes 2}(\psi')^{\otimes 2} M_\varepsilon^{\otimes 2}]]  =  \Tr[M_\varepsilon^2](\Tr[M_\varepsilon^2 \phi_\varepsilon]  -f_\varepsilon)  \\  + \Tr[M_\varepsilon^2 \phi_\varepsilon](4f_\varepsilon - \Tr[M_\varepsilon^2 \psi_\varepsilon] - 2\Tr[M_\varepsilon^2 \phi_\varepsilon]) +2\Tr[M_\varepsilon^2 \psi_\varepsilon] + \Tr[M_\varepsilon^4 \phi_\varepsilon] - f_\varepsilon^2 \notag \\
                    - 2\text{Re}(\langle \phi_\varepsilon | M_\varepsilon^3 | \psi_\varepsilon \rangle\langle \psi_\varepsilon | M_\varepsilon | \phi_\varepsilon \rangle)\ . \notag
                \end{align}
                Dropping all of the non-positive terms proves the lemma.

                \textbf{\emph{Bound on term 3.}}   This follows from steps similar to those of term 2 and we omit the details.
                
                \textbf{\emph{Bound on term 4.}}    \begin{align}
                    \E[| q |^2 | g |^2] & = \E[\alpha^2(1-\alpha^2) \beta^4] f_\varepsilon \E[\Tr[\psi_\varepsilon M_\varepsilon \phi' M_\varepsilon]]\\
                    & = \frac{(s_b+2)(s_b+1)(s_a+1)}{D}f_\varepsilon \Tr[\psi_\varepsilon M_\varepsilon P_{\phi'} M_\varepsilon]\\
                    & = \frac{(s_b+2)(s_b+1)(s_a+1)}{D}f_\varepsilon(\Tr[M_\varepsilon^2 \psi_\varepsilon] - f_\varepsilon)\\
                    & \leq \frac{(s_b+2)(s_b+1)(s_a+1)}{D}\ .
                \end{align}
                
                \textbf{\emph{Bound on term 5.}}    This follows from steps similar to those of term 4 and we omit the details.
                
                \textbf{\emph{Bound on term 6.}} \begin{align}
                    \E[| q |^2 | \ell |^2] & = \E[\alpha^2(1-\alpha^2)\beta^2(1-\beta^2)]\cdot f_\varepsilon  \cdot \E[\Tr[\phi' M_\varepsilon \psi' M_\varepsilon]]\\
                    & = \frac{(s_a+1)(s_b+1)}{D}f_\varepsilon \Tr[P_{\phi'} M_\varepsilon P_{\psi'} M_\varepsilon]\\
                    & = \frac{(s_a+1)(s_b+1)}{D}f_\varepsilon \left( \Tr[M_\varepsilon^2(\id_{d_\varepsilon} - \psi_\varepsilon - \phi_\varepsilon)] +f_\varepsilon \right)\\
                    & \leq \frac{(s_a+1)(s_b+1)}{D} \Tr[M_\varepsilon^2]\ .
                \end{align}
         \end{proof}

            \begin{lemma} [Bounding the degree-$4$ terms]
            \label{lem:quartic}
                \begin{align}
                    \E[q \cdot \ell \cdot  g^* \cdot h^* + q^* \cdot \ell^* \cdot g \cdot h] \leq \frac{2(s_a+1)(s_b+1)}{D}\ .
                \end{align}
            \end{lemma}
            \begin{proof}
                \begin{align}
                    | \E[q \cdot \ell \cdot g^* \cdot h^*] | & = |\E\big[\alpha^2(1-\alpha^2)\beta^2(1-\beta^2) \cdot \langle \phi_\varepsilon | M_\varepsilon | \psi_\varepsilon \rangle \cdot  \langle \psi_\varepsilon | M_\varepsilon | \phi' \rangle \langle \phi' | M_\varepsilon | \psi' \rangle \langle \psi' | M_\varepsilon | \phi_\varepsilon \rangle\big] |\\
                    & \leq \frac{(s_a+1)(s_b+1)}{D} | \langle \psi_\varepsilon | M_\varepsilon P_{\phi'} M_\varepsilon P_{\psi'} M_\varepsilon | \phi_\varepsilon \rangle |\\
                    & \leq \frac{(s_a+1)(s_b+1)}{D} \Vert M_\varepsilon P_{\phi'} M_\varepsilon P_{\psi'} M_\varepsilon \Vert \\
                    & \leq \frac{(s_a+1)(s_b+1)}{D}\ ,
                \end{align}
                where the last inequality follows from $\Vert M_\varepsilon \Vert = | P_{\phi'} \Vert = \Vert P_{\psi'} \Vert = 1$. The same steps hold for the complex conjugate.
            \end{proof}

\subsection{Lower Bound}
We now prove the claimed lower bounds in Theorem~\ref{thm:blf_ub}. The first, $\Omega(1/\varepsilon^2)$, follows from lemma 13 in~\cite{anshu2022distributed}. We restate their argument here slightly adapted to our setup.

\begin{lemma}
    Say there is an algorithm acting on $\rho^{\otimes k} \otimes \sigma^{\otimes k}$ that outputs an estimate of $\Tr[M\rho M \sigma]$ to accuracy $\varepsilon$ with high probability. Then, $k=\Omega(1/\varepsilon^2)$.
\end{lemma}
\begin{proof}
    Let $\ket{0}$ be such that $M \ket{0} = \pm 1 \ket{0}$, which exists since $\Vert M \Vert = 1 $. Let $\ket{1}$ be an eigenvector of $M$ orthogonal to $\ket{0}$. Consider the states
    \begin{align}
        \ket{\psi_0}  = \sqrt{\frac{1}{2}-\varepsilon} \ket{0} + \sqrt{\frac{1}{2}+\varepsilon}\ket{1}\ ,\quad 
        \ket{\psi_1}  = \sqrt{\frac{1}{2}+\varepsilon} \ket{0} + \sqrt{\frac{1}{2}-\varepsilon}\ket{1}\ .
    \end{align}
    Now, say that Alice is given $k$ copies of $\ket{\psi_0}$ or $\ket{\psi_1}$ and Bob is given $k$ copies of $\ket{0}$. We have that
    \begin{align}
        | \langle \psi_0 | M | 0 \rangle |^2  = \frac{1}{2}-\varepsilon, \quad 
        | \langle \psi_1 | M | 0 \rangle |^2 = \frac{1}{2}+\varepsilon \ .
    \end{align}
    Thus, if Alice and Bob can estimate $\vert \langle \psi | M | \phi \rangle |^2$ to accuracy $\varepsilon$, they can distinguish between these two states. However, the fidelity between these states is given by
    \begin{align}
        F(\psi_0^{\otimes k}, \psi_1^{\otimes k}) = (1-4\varepsilon^2)^k \geq 1-4k\varepsilon^2\ ,
    \end{align}
    which implies that $k = \Omega(1/\varepsilon^2)$.
\end{proof}
Before proving the second lower bound, we give some intuition. Say that $M \succeq 0$.If $\ket{\psi}$ and $\ket{\phi}$ are drawn independently from the Haar measure, then
\begin{align}
    \E[\Tr[M \psi M \phi]] & =\Tr[M^2]/d^2.
\end{align}
If they are identical, that is $\psi = \phi$ always, then instead the expected value is
\begin{align}
    \E[\Tr[M \psi M \psi]] & = \frac{1}{d(d-1)}\left( \Tr[M^2] + \Tr[M]^2 \right)\ .
\end{align}
The difference between these two is roughly $\Tr[M]^2 / d^2$. Thus, if Alice and Bob are able to estimate $\vert \langle \psi | M | \phi \rangle \vert^2$ to this order, then they can solve $\DIPE$ (that we defined as Problem~\ref{def:DIPE}). However, $\Tr[M]^2/d^2$ may be quite small. Restricted to the support of $M_\varepsilon$, we instead have that $\varepsilon^2/4 \leq \Tr[M_\varepsilon]^2/d_{\varepsilon}^2 \leq 1$.
We will show that if Alice and Bob can estimate $\vert \langle \phi \vert M_\varepsilon \vert \psi \rangle \vert^2$, then they can solve the following decision problem, which is known to require $k=\Omega(\sqrt{\dim \mathcal{H}}/\varepsilon)$ samples~\cite{anshu2022distributed} (when $\varepsilon \leq 0.01$).
\begin{definition}[Inner product estimation, $\varepsilon$ decision version]\label{def:ip_decision}
    Let $\varepsilon >0$. Alice and Bob are given $k$ copies each of pure states $\ket{\psi},\ket{\phi} \in \mathcal{H}$, for some finite dimensional Hilbert space $\mathcal{H}$. They are asked to decide, using $\LOCC$, which of the following two scenarios they are in:
    \begin{itemize}
        \item $\YES$ instance: 
        \begin{align}
            \ket{\psi} = \sqrt{1-\eps}e^{i\theta}\ket{0}+\sqrt{\eps}\ket{\chi}\ , \quad \ket{\phi} = \sqrt{1-\eps}e^{i\theta'}\ket{0}+\sqrt{\eps}\ket{\chi}\ ,
        \end{align}
        where $\ket{\chi}$ is drawn from the Haar measure on the orthogonal complement of $\ket{0}$ and $\theta$ and $\theta'$ are independently and uniformly drawn from $[0,2\pi)$.

        \item $\NO$ instance:
        \begin{align}
            \ket{\psi} = \sqrt{1-\eps}e^{i\theta}\ket{0}+\sqrt{\eps}\ket{\chi}\ , \quad \ket{\phi} = \sqrt{1-\eps}e^{i\theta'}\ket{0}+\sqrt{\eps}\ket{\varphi}\ ,
        \end{align}
        where $\ket{\chi}$ and $\ket{\varphi}$ are drawn independently from the Haar measure on the orthogonal complement of $\ket{0}$ and $\theta$ and $\theta'$ are independently and uniformly drawn from $[0,2\pi)$.
    \end{itemize}
\end{definition}

We will require the two following technical lemmas in proving the lower bound. Here Lipschitz constants are taken to be with respect to the Euclidean norm $\Vert \cdot \Vert_2$.

\begin{fact}[Levy's Lemma~\cite{brannan2021alice}]\label{lem:levy}
                If $f:\mathbb{S}^d \rightarrow \mathbb{R}$ is $\lambda$-Lipschitz, then
                \begin{align}
                    \Pr_u[\vert f(u) - \E[f(u)] \vert \geq t ] \leq 2 e^{-\frac{dt^2}{2\lambda^2}}\ ,
                \end{align}
                where $u$ is drawn from the Haar measure on $\mathbb{S}^d$.
\end{fact}
\noindent The functions we care about are $\Tr[M(\cdot)]$ defined on the sphere $\mathbb{C}^d$, which can be identified with~$\mathbb{S}^{2d}$.
\begin{lemma}
    $\Tr[M\chi]$ is $2$-Lipschitz in $\chi$ with respect to $\Vert \cdot \Vert_2$.
\end{lemma}
\begin{proof}
       \begin{align}
           \vert  \langle \psi \vert M \vert \psi \rangle - \langle \phi \vert M \vert \phi \rangle \vert & = \vert (\langle \psi \vert - \langle \phi \vert) M \vert \psi \rangle - \langle \phi \vert M (\vert \phi \rangle - \vert \psi \rangle) \vert\\
           & \leq \vert (\langle \psi \vert - \langle \phi \vert) M \vert \psi \rangle \vert + \vert \langle \phi \vert M (\vert \phi \rangle - \vert \psi \rangle) \vert\\
           & \leq 2\Vert \ket{\psi} - \ket{\phi}\Vert\ ,
       \end{align}
       where the final inequality follows from $\Vert M \Vert = 1$ and Cauchy-Schwarz.
\end{proof}

\begin{lemma}\label{lem:bil_lb}
    Let $c<1$ be a small constant. Suppose a protocol acting on $\rho^{\otimes k}\otimes \sigma^{\otimes k}$ via $\LOCC$ outputs an estimate of $\Tr[M\rho M\sigma]$ to accuracy $c\varepsilon$ with high probability. Then, $k = \Omega(\| M_\varepsilon \|_1 / \varepsilon \sqrt{d_\varepsilon})$.
\end{lemma}
\begin{proof}
   Let $\ket{0}$ be such that $M \ket{0} = \ket{0}$. Such a vector exists by considering either $M$ or $-M$, and we choose the appropriate sign without loss of generality. Let $W = \{\ket{\psi} \in \mathbb{C}^d \ \vert \ \langle 0 \vert \psi\rangle = 0\} \cap \text{Im}P_\varepsilon$ be the subspace spanned by all eigenvectors with eigenvalue of magnitude at least $\eps/2$ other than $\ket{0}$. Let $\tilde{M_\varepsilon} = M_\varepsilon\vert_W$ be the restriction of $M$ to this subspace. Without loss of generality we assume that $\tilde{M_\varepsilon}$ is a (semi)definite matrix. Indeed, let $\tilde{M_\varepsilon} = \tilde{M_\varepsilon}^+ - \tilde{M_\varepsilon}^-$ then $\Vert \tilde{M_\varepsilon} \Vert^2_2 = \Vert \tilde{M_\varepsilon}^+ \Vert_2^2 + \Vert \tilde{M_\varepsilon}^- \Vert_2^2$. Of course, this means that $\Vert \tilde{M_\varepsilon}^+\Vert_2 \geq \Vert \tilde{M_\varepsilon}\Vert_2/\sqrt{2}$ or $\Vert \tilde{M_\varepsilon}^-\Vert_2 \geq \Vert \tilde{M_\varepsilon} \Vert_2/\sqrt{2}$. Thus, we restrict ourselves further and consider the subspace, again labeled by $W$, in the support of the operator with a larger~norm.

   We will now show that estimating $\vert \langle \phi \vert M \vert \psi \rangle \vert^2$ suffices to solve problem~\ref{def:ip_decision} with high probability where $\mathcal{H} = \mathbb{C}\ket{0} \oplus W$. Let $\ell := \dim W$ and note that, by construction, $\Tr[\tilde{M_\eps}] = \Omega(\Vert M_\eps \Vert_2)$. We assume that $\varepsilon \geq \max\{56, 51/100c\}/\sqrt{\ell}$. This is no issue: if $\eps = o(1/\sqrt{d})$, then the lower bound $k= \Omega(1/\varepsilon^2)$ dominates since $\Vert \tilde{M_\eps} \Vert_2 /\sqrt{\eps} \leq \sqrt{\ell/\eps} = o(1/\varepsilon^2)$.   
   Set the precision parameter in problem~\ref{def:ip_decision} to be $\delta:= \varepsilon\ell/200\Tr[\tilde{M_\varepsilon}]$ and define $\vartheta := \theta - \theta'$. Since $\Tr[\tilde{M_\varepsilon}] \geq \varepsilon\ell/2$, $\delta \leq 0.01$ as required. Let $f$ be Alice and Bob's estimate of $\vert \langle \phi \vert M \vert \psi \rangle \vert^2$, and say that it is within $c\cdot\varepsilon$ of the true value with probability at least 0.99. Then, if $f$ lies in the range $[(1-\delta)^2 -3c\cdot\varepsilon, (1-\delta)^2 +c\cdot\varepsilon]$, they $\mathsf{REJECT}$. Otherwise, they $\mathsf{ACCEPT}$. We split the proof into the two cases.

 \textbf{ $\YES$ instance} : in this case
       \begin{align}
           \vert \langle \phi \vert M \vert \psi \rangle \vert^2 & =(1-\delta)^2 + \delta^2 \Tr[M\chi]^2 + 2\delta(1-\delta)\cos \vartheta \Tr[M\chi]\ .
       \end{align}
       Thus,
       \begin{align}
           &\Pr[f \in [(1-\delta)^2 - \frac{\delta}{8}, (1-\delta)^2 + x]]\\
           &\leq \Pr[ \delta^2 \Tr[M\chi]^2 + 2\delta(1-\delta)\cos \vartheta \Tr[M\chi] \in \left[-\frac{\delta}{8}-c\varepsilon, \frac{\delta}{8} + {c\varepsilon}\right] + 0.01 \ .
       \end{align}
Then, using fact~\ref{lem:levy}, it holds that
       \begin{align}
           \Pr\left[ \vert \Tr[M\chi] - \E_\chi[\Tr[M\chi]] \vert \leq \frac{7}{\sqrt{\ell}} \right] > 0.99\ .  
       \end{align}
We have that $\Tr[\tilde{M_\eps}]/4\ell \geq \varepsilon/8 \geq 7/\sqrt{\ell}$, where the last inequality follows from the assumption that $\eps$ is at least $56/\sqrt{\ell}$. The above bounds imply that we can take $\Tr[M\chi]^2 \leq 25\Tr[\tilde{M_\varepsilon}]^2/16\ell^2$. We are then interested in the probability
    \begin{align}
        \Pr\left[ 2\delta(1-\delta)\frac{3{\Tr[\tilde M_\varepsilon]}}{4\ell}\cos\vartheta  \in \left[-\frac{\delta}{8}-c\varepsilon -\delta^2\frac{25\Tr[\tilde{M_\varepsilon}]^2}{16\ell^2}, \frac{\delta}{8}+c\varepsilon - \delta^2\frac{\Tr[\tilde{M_\varepsilon}]^2}{\ell^2} \right] \right]\ .
    \end{align}
    Now, by construction $\delta \leq 0.01$ and thus $2(1-\delta)\geq 99/50$. It follows that
    \begin{align}
        2\delta(1-\delta)\frac{3{\Tr[\tilde M_\varepsilon]}}{4\ell} \geq \frac{297}{40000}\varepsilon\ .
    \end{align}
    Then, the above probability is upper bounded by
    \begin{align}
        \Pr\left[ \cos\vartheta \in \left[-\frac{1}{2}-\frac{25\varepsilon}{4752}, \frac{1}{2} + \frac{\varepsilon}{297}\right]\right]
        & \leq \Pr\left[ \cos\vartheta \in \left[-0.51, 0.51 \right] \right] \leq 0.34\ .
    \end{align}
    In total, they accept in the $\YES$ case with probability at least $0.66$.

\textbf{$\NO$ instance}: in this case
    \begin{align}
        \vert \langle \varphi \vert M \vert \chi \rangle \vert^2 & = (1-\delta)^2 + \delta^2\Tr[M\varphi M \chi] + 2\delta(1-\delta)\text{Re}\left(e^{i\vartheta}\langle \varphi \vert M \vert \chi \rangle \right)\ .
    \end{align}
    By symmetry, $\E[\text{Re}\left(e^{i\vartheta}\langle \varphi \vert M \vert \chi \rangle \right)]= 0$. Now, fixing $\ket{\varphi}$, $\vert \langle \varphi \vert M \vert \chi \rangle \vert$ is $1$-Lipschitz in $\ket{\chi}$. This implies that $\vert\langle \varphi \vert M \vert \chi \rangle\vert \leq 8/\sqrt{\ell}$ with probability at least $0.999$. This then implies that $\Tr[M\varphi M\chi] \leq 84/\ell$. Then, we have that
    \begin{align}
         \left\vert  \vert \langle \varphi \vert M \vert \chi \rangle \vert^2 - (1-\delta)^2 \right\vert \leq \frac{\delta}{\sqrt{\ell}}\left( \frac{84}{\sqrt{\ell}} + 16 \right)\ .
    \end{align}
    Further, $\delta/\sqrt{\ell} \leq 1/100\sqrt{\ell}$. We assume that $84/\sqrt{\ell} < 1$ (as otherwise the lower bound is constant anyways) and obtain
    \begin{align}
         \left\vert  \vert \langle \varphi \vert M \vert \chi \rangle \vert^2 - (1-\delta)^2 \right\vert \leq \frac{17}{100\sqrt{\ell}}\ .
    \end{align}
    We require that this is less then $c\varepsilon/3$. This requires that $\varepsilon = \Omega(1/\sqrt{\ell})$, but, as previously stated, we assume this as otherwise the other lower bound kicks in. The total probability for this instance to be rejected is then at least $0.9$.     Thus, Alice and Bob are able to solve problem~\ref{def:ip_decision} with high probability. Then, it must be that $k=\Omega(\sqrt{\ell}/\delta) = \Omega(\Tr[\tilde{M_\varepsilon}]/\varepsilon\sqrt{\ell})$.
\end{proof}

\begin{corollary}
     Let $c<1$ be a small constant. If there is a protocol acting on $\rho^{\otimes k}\otimes \sigma^{\otimes k}$ via $\LOCC$ that outputs an  $c\varepsilon$-accurate estimate of $\Tr[M\rho M\sigma]$  with high probability, then $k = \Omega(\| M_\varepsilon \|_2 / \sqrt{\varepsilon})$.
\end{corollary}
\begin{proof}
    Lemma~\ref{lem:bil_lb} yields a lower bound of $k=\Omega(\Tr[\tilde{M_\varepsilon}]/\varepsilon\sqrt{\ell})$. Using $\Tr[\tilde{M_\varepsilon}] \geq \Vert \tilde{M_\varepsilon} \Vert_2^2$ and $\Vert \tilde{M_\varepsilon} \Vert_2 \geq \sqrt{\varepsilon \ell/2}$, we arrive at our claimed lower bound of $k=\Omega(\Vert M_\varepsilon \Vert_2 / \sqrt{\varepsilon})$.
\end{proof}

\bibliographystyle{alpha}
\bibliography{ref}

\end{document}